\documentclass[a4paper,11pt]{amsart}
\usepackage{amsmath,amssymb,amsthm}
\usepackage{epsfig}
\usepackage{graphicx}

\usepackage[hmargin=2cm,vmargin=2.0cm]{geometry}

\usepackage{hyperref}
\usepackage{nameref,zref-xr}                    

\makeatletter
\def\Ginclude@eps#1{%
 \message{<#1>}%
  \bgroup
  \def\@tempa{!}%
  \dimen@\Gin@req@width
  \dimen@ii.1bp%
  \divide\dimen@\dimen@ii
  \@tempdima\Gin@req@height
  \divide\@tempdima\dimen@ii
    \includegraphics{#1}%
  \egroup}
\makeatother

\renewcommand{\>}{\rangle}
\newcommand{\<}{\langle}

\newcommand{\C}{\operatorname{\mathfrak{C}}}

\newcommand{\CP}{\mathbb{C}\mathbb{P}^1}

\newcommand{\del}{\partial}

\newcommand{\Res}{\operatorname{Res}}

\newcommand{\Z}{\mathbb Z}
\newcommand{\mbC}{\mathbb C}

\newcommand{\mbZ}{\mathbb Z}
\newcommand{\mcN}{\mathcal N}
\newcommand{\mcA}{\mathcal A}

\newcommand{\half}{\frac{1}{2}}
\newcommand{\halfz}{\frac{z}{2}}
\newcommand{\halfdx}{\frac{\d_x}{2}}

\newcommand{\hcA}{\widehat{\mathcal A}}

\newcommand{\eps}{\varepsilon}

\def\d{\partial}
\def\CP1{\mathbb{C}\mathbb{P}^1}

\newcommand{\oM}{\overline{\mathcal M}}

\newcommand{\og}{\overline g}
\newcommand{\oh}{\overline h}

\newcommand{\hLambda}{\widehat\Lambda}

\def\M{{\mathcal{M}}}

\def\oM{{\overline{\mathcal{M}}}}

\def\Z{{\mathbb Z}}

\def\C{{\mathbb C}}
\def\Q{{\mathbb Q}}
\def\d{{\partial}}

\def\t{\tilde}

\newcommand{\tu}{\widetilde u}
\newcommand{\os}{\overline s}
\newcommand{\tQ}{\widetilde Q}
\newcommand{\tq}{\widetilde q}
\newcommand{\tL}{\widetilde L}
\newcommand{\cB}{\mathcal B}
\newcommand{\Coef}{\mathrm{Coef}}
\newcommand{\DR}{\mathrm{DR}}

\newtheorem{theorem}{Theorem}[section]

\newtheorem{lemma}[theorem]{Lemma}

\newtheorem{rem}[theorem]{Remark}
\newtheorem{ex}[theorem]{Example}

\newenvironment{example}{\begin{ex}\rm}{\qee\end{ex}}
\newcommand{\qee}{\mbox{\hspace{0.2mm}}\hfill$\triangle$}

\numberwithin{equation}{section}

\title{Recursion relations for Double Ramification Hierarchies}

\author{Alexandr Buryak}
\address{A.~Buryak:\newline Department of Mathematics, ETH Zurich, \newline Ramistrasse 101 8092, HG G 27.1, Zurich, Switzerland}
\email{buryaksh\_at\_gmail.com}

\author{Paolo Rossi}
\address{P.~Rossi:\newline IMB, UMR 5584 CNRS, Universit\'e de Bourgogne,\newline 9, avenue Alain Savary, 21078 Dijon Cedex, France}
\email{paolo.rossi\_at\_u-bourgogne.fr}

\begin{document}

\begin{abstract}
In this paper we study various properties of the \emph{double ramification hierarchy}, an integrable hierarchy of hamiltonian PDEs introduced in \cite{Bur14} using intersection theory of the double ramification cycle in the moduli space of stable curves. In particular, we prove a recursion formula that recovers the full hierarchy starting from just one of the Hamiltonians, the one associated to the first descendant of the unit of a cohomological field theory. Moreover, we introduce analogues of the topological recursion relations and the divisor equation both for the hamiltonian densities and for the string solution of the double ramification hierarchy. This machinery is very efficient and we apply it to various computations for the trivial and Hodge cohomological field theories, and for the $r$-spin Witten's classes. Moreover we prove the Miura equivalence between the double ramification hierarchy and the Dubrovin-Zhang hierarchy for the Gromov-Witten theory of the complex projective line (extended Toda hierarchy).
\end{abstract}

\maketitle

\tableofcontents

\markboth{A. Buryak, P. Rossi}{Recursion relations for DR hierarchies}

\section{Introduction}
In a recent paper, \cite{Bur14}, one of the authors, inspired by ideas from symplectic field theory~\cite{EGH00}, has introduced a new integrable hierarchy of PDEs associated to a given cohomological field theory. The construction makes use of the intersection numbers of the given cohomological field theory with the double ramification cycle, the top Chern class of the Hodge bundle and psi-classes on the moduli space of stable Riemann surfaces $\oM_{g,n}$. Since the top Chern class of the Hodge bundle vanishes outside of the moduli space of stable curves of compact type, one can use Hain's formula~\cite{H11} to express the double ramification cycle in computations and, in particular, the consequent polynomiality of the double ramification cycle with respect to ramification numbers.\\

In \cite{Bur14} the author further conjectures, guided by the examples of the trivial and the Hodge cohomological field theories (which give the KdV and the ILW hierarchies, respectively) that the double ramification hierarchy is Miura equivalent to the Dubrovin-Zhang hierarchy associated to the same cohomological field theory via the construction described, for instance, in~\cite{DZ05}.\\

In this paper, after defining some natural hamiltonian densities for the double ramification hierarchy, using results from \cite{BSSZ12} we derive a series of equations for such densities. Some of these equations are reminiscent of the topological recursion relations and the divisor equation in Gromov-Witten theory~\cite{KM94,Get97}, but also of their analogues from symplectic field theory \cite{FR10,Ros12}. The dilaton recursion of Theorem \ref{dilaton}, in pariticular, is sufficient to recover the full hierarchy of the hamiltonian densities, starting just from one Hamiltonian (the one associated with the first descendant of the unit of the cohomological field theory). We apply this technique to compute explicit formulae for the double ramification hierarchy of the $r$-spin Witten's classes for $r=3,4$ and, in particular, we conjecture explicit formulae for the Miura transformations that should link such hierarchy to the Dubrovin-Zhang hierarchy.\\

In the second part we focus instead on the string solution of the double ramification hierarchy (see~\cite{Bur14}) and prove that the divisor equation for the Hamiltonians implies the divisor equation for the string solution. Our main application of this fact is a proof of the Miura equivalence described above in the case of the Gromov-Witten theory of the complex projective line. Consider the Gromov-Witten theory of $\CP1$ and the corresponding cohomological field theory. We use $u^\alpha$ as the variables of the double ramification hierarchy for $\CP1$ and $w^\alpha$ as the variables of the ancestor Dubrovin-Zhang hierarchy for $\CP1$. 
\begin{theorem}\label{theorem:dr for cp1}
The double ramification hierarchy for $\CP1$ is related to the ancestor Dubrovin-Zhang hierarchy for~$\CP1$~by the Miura transformation
\begin{gather}\label{eq:Miura}
u^\alpha(w)=\frac{e^{\frac{\eps}{2}\d_x}-e^{-\frac{\eps}{2}\d_x}}{\eps\d_x}w^\alpha.
\end{gather}
\end{theorem}

\vspace{1cm}

\noindent{\bf Acknowledgements.}\\
We would like to thank Boris Dubrovin, Rahul Pandharipande, Sergey Shadrin and Dimitri Zvonkine for useful discussions. P. R. was partially supported by a Chaire CNRS/Enseignement superieur 2012-2017 grant. A. B. was supported by grant ERC-2012-AdG-320368-MCSK in the group of R.~Pandharipande at ETH Zurich, by the Russian Federation Government grant no. 2010-220-01-077 (ag. no. 11.634.31.0005), the grants RFFI 13-01-00755 and NSh-4850.2012.1.

Part of the work was completed during the visit of A.B to the University of Burgundy in 2014 and during the visit of P.R. to the Forschungsinstitut f\"ur Mathematik at ETH Z\"urich in 2014.

\vspace{0.5cm}


\section{The double ramification hierarchy}

In this section we briefly recall the main definitions in \cite{Bur14}. The double ramification hierarchy is a system of commuting Hamiltonians on an infinite dimensional phase space that can be heuristically thought of as the loop space of a fixed vector space. The entry datum for this construction is a cohomological field theory  in the sense of Kontsevich and Manin \cite{KM94}. Denote by $c_{g,n}\colon V^{\otimes n} \to H^{even}(\oM_{g,n};\mbC)$ the system of linear maps defining the cohomological field theory, $V$ its underlying $N$-dimensional vector space, $\eta$ its metric tensor and $e_1$ the unit of the cohomological field theory.

\subsection{The formal loop space} 

The loop space of $V$ will be defined somewhat formally by describing its ring of functions. Following \cite{DZ05} (see also \cite{Ros09}), let us consider formal variables $u^\alpha_p$, $\alpha=1,\ldots,N$, $p=0,1,2,\ldots$, associated to a basis $e_1,\ldots,e_N$ for $V$. Always just at a heuristic level, the variable $u^\alpha:=u^\alpha_0$ can be thought of as the component $u^\alpha(x)$ along $e_\alpha$ of a formal loop $u:S^1\to V$, where $x$ the coordinate on $S^1$, and the variables $u^\alpha_{x}:=u^\alpha_1, u^\alpha_{xx}:=u^\alpha_2,\ldots$ as its $x$-derivatives. We then define the ring $\mathcal{A}$ of \emph{differential polynomials} as the ring of polynomials $f(u;u_x,u_{xx},\ldots)$ in the variables~$u^\alpha_i, i>0$, with coefficients in the ring of formal power series in the variables $u^\alpha=u^\alpha_0$. We can differentiate a differential polynomial with respect to $x$ by applying the operator $\partial_x := \sum_{i\geq 0} u^\alpha_{i+1} \frac{\partial}{\partial u^\alpha_i}$ (in general, we use the convention of sum over repeated greek indices, but not over repeated latin indices). Finally, we consider the quotient $\Lambda$ of the ring of differential polynomials first by constants and then by the image of $\partial_x$, and we call its elements \emph{local functionals}. A local functional which is the equivalence class of $f=f(u;u_x,u_{xx},\ldots)$ will be denoted by $\overline{f}=\int f dx$. Strictly speaking, in order to obtain the ring of functions for our formal loop space, we must consider a completion of the symmetric tensor algebra of the space of local functionals whose elements correspond to multiple integrals on multiple copies of the variable $x$ of differential polynomials of multiple copies of the variables $u^\alpha_i$, but we will not really use this in the paper.\\

Differential polynomials and local functionals can also be described using another set of formal variables, corresponding heuristically to the Fourier components $p^\alpha_k$, $k\in\Z$, of the functions $u^\alpha=u^\alpha(x)$. Let us, hence, define a change of variables
$$
u^\alpha_j = \sum_{k\in \Z} (i k)^j p^\alpha_k e^{i k x},
$$
which allows us to express a differential polynomial $f$ as a formal Fourier series in $x$ where the coefficient of $e^{i k x}$ is a power series in the variables $p^\alpha_j$ (where the sum of the subscripts in each monomial in~$p^\alpha_j$ equals~$k$). Moreover, the local functional~$\overline{f}$ corresponds to the constant term of the Fourier series of $f$.\\

Let us describe a natural class of Poisson brackets on the space of local functionals. Given a matrix of differential operators of the form $K^{\mu\nu} = \sum_{j\geq 0} K^{\mu\nu}_j \partial_x^j$, where the coefficients $K^{\mu\nu}_j$ are differential polynomials and the sum is finite, we define
$$
\{\overline{f},\overline{g}\}_{K}:= \int \left(\frac{\delta \overline{f}}{\delta u^\mu} K^{\mu \nu} \frac{\delta \overline{g}}{\delta u^\nu}\right) dx,
$$
where we have used the variational derivative $\frac{\delta \overline{f}}{\delta u^\mu}:=\sum_{i\geq 0} (-\partial_x)^i \frac{\partial f}{\partial u^\mu_i}$. Imposing that such bracket satisfies the anti-symmetry and the Jacobi identity will translate, of course, into conditions for the coefficients~$K^{\mu \nu}_j$. An operator that satisfies such conditions will be called hamiltonian. We will also define a Poisson bracket between a differential polynomial and a local functional as
$$
\{f,\overline{g}\}_{K}:= \sum_{i\geq 0} \frac{\partial f}{\partial u^\mu_i} \partial_x^i \left( K^{\mu \nu} \frac{\delta \overline{g}}{\delta u^\nu}\right) dx.
$$
A standard example of a hamiltonian operator is given by $\eta \partial_x$. The corresponding Poisson bracket, heavily used in what follows, also has a nice expression in terms of the variables $p^\alpha_k$:
$$\{p^\alpha_k, p^\beta_j\}_{\eta \partial_x} = i \eta^{\alpha \beta} k \delta_{k+j,0}.$$\\

Finally, we will need to consider extensions of the spaces $\mathcal{A}$ and $\Lambda$ of differential polynomials and local functionals. First, let us introduce a grading $\deg u^\alpha_i = i$  and a new variable $\eps$ with $\deg\eps = -1$. Then $\hcA^{[k]}$ and $\hLambda^{[k]}$ are defined, respectively, as the subspaces of degree $k$ of $\hcA:= \mathcal{A}  \otimes \C[[\eps]]$ and of~$\hLambda:=\Lambda \otimes \C[[\eps]] $. Their elements will still be called differential polynomials and local functionals. We can also define Poisson brackets as above, starting from hamiltonian operators $K^{\mu\nu} = \sum_{i,j\geq 0} K^{\mu\nu}_{ij} \eps^i \partial_x^j$, where $K^{\mu\nu}_{ij}$ are differential polynomials of degree $i-j+1$. The corresponding Poisson bracket will then have degree $1$. \\

A hamiltonian system of PDEs is a system of the form
\begin{gather}\label{eq:Hamiltonian system}
\frac{\partial u^\alpha}{\partial \tau_i} = K^{\alpha\mu} \frac{\delta\overline{h}_i}{\delta u^\mu}, \ \alpha=1,\ldots,N ,\  i=1,2,\ldots,
\end{gather}
where $\overline{h}_i\in \hLambda^{[0]}$ are local functionals with the compatibility condition $\{\oh_i,\oh_j\}_K=0$ for $i,j\geq 1$.

\subsection{The double ramification hierarchy} 

Given a cohomological field theory 
$$
c_{g,n}\colon V^{\otimes n} \to H^{even}(\oM_{g,n};\mbC),
$$
we define hamiltonian densities of the double ramification hierarchy as the following generating series:
\begin{equation}\label{density}
\begin{split}
g_{\alpha,d}:=&\sum_{\substack{g\ge 0,n\ge 1\\2g-1+n>0}}\frac{(-\eps^2)^g}{n!}\times\\
&\times\sum_{a_1,\ldots,a_n\in\mbZ}\left(\int_{\DR_g\left(-\sum a_i,a_1,\ldots,a_n\right)}\lambda_g\psi_1^d c_{g,n+1}(e_\alpha\otimes e_{\alpha_1}\otimes\ldots\otimes e_{\alpha_n})\right)p^{\alpha_1}_{a_1}\ldots p^{\alpha_n}_{a_n}e^{ix\sum a_i},
\end{split}
\end{equation}
for $\alpha=1,\ldots,N$ and $d=0,1,2,\ldots$. Here $\DR_g\left(a_1,\ldots,a_n\right) \in H^{2g}(\oM_{g,n};\Q)$ is the double ramification cycle, $\lambda_g$ is the $g$-th Chern class of the Hodge bundle and $\psi_i$ is the first Chern class of the tautological bundle at the $i$-th marked point.\\

The above expression can be uniquely written as a differential polynomial in $u^\mu_i$ in the following way. In genus $0$ we have
\begin{gather}\label{eq:genus zero DR}
\DR_g(a_1,\ldots,a_n)=[\oM_{0,n}].
\end{gather}
In higher genera $g>0$ Hain's formula~\cite{H11} together with the result of~\cite{MW13} imply that
\begin{gather}\label{eq:Hain's formula}
\left.\DR_g(a_1,\dotsc,a_n)\right|_{\M^{ct}_{g,n}}=\frac{1}{g!}\left(\sum_{j=1}^n \frac{a_j^2 \psi^{\dagger}_j}{2} - \sum_{\substack{J \subset \left\lbrace 1,\dotsc,n\right\rbrace  \\ \left| J \right| \geq 2}} \left(\sum_{i,j \in J, i<j} a_i a_j\right) \delta_0^J - \frac{1}{4} \sum_{J \subset \left\lbrace 1,\dotsc,n\right\rbrace} \sum_{h=1}^{g-1} a_J^2 \delta_h^J \right)^g,
\end{gather}
where $\M_{g,n}^{ct}$ is the moduli space of curves of compact type, $\psi^\dagger_j$ denotes the $\psi$-class that is pulled back from~$\oM_{g,1}$, the integer $a_J$ is the sum $\sum_{j\in J} a_j$ and the class $\delta_h^J$ represents the divisor whose generic point is a nodal curve made of one smooth component of genus $h$ with the marked points labeled by the list $J$ and of another smooth component of genus $g-h$ with the remaining marked points, joined at a separating node. From formulae~\eqref{eq:genus zero DR},~\eqref{eq:Hain's formula} and the fact that $\lambda_g$ vanishes on $\oM_{g,n} \setminus \M_{g,n}^{ct}$ it follows that the integral
\begin{gather}\label{integral for the density}
\int_{\DR_g\left(-\sum a_i,a_1,\ldots,a_n\right)}\lambda_g\psi_1^d c_{g,n+1}(e_\alpha\otimes e_{\alpha_1}\otimes\ldots\otimes e_{\alpha_n})
\end{gather}
is a polynomial in $a_1,\ldots,a_n$ homogeneous of degree $2g$. Denote it by 
$$
P_{\alpha,d,g;\alpha_1,\ldots,\alpha_n}(a_1,\ldots,a_n)=\sum_{\substack{b_1,\ldots,b_n\ge 0\\b_1+\ldots+b_n=2g}} P_{\alpha,d,g;\alpha_1,\ldots,\alpha_n}^{b_1,\ldots,b_n}a_1^{b_1}\ldots a_n^{b_n}.
$$
Then we have
$$
g_{\alpha,d}=\sum_{\substack{g\ge 0,n\ge 1\\2g-1+n>0}}\frac{\eps^{2g}}{n!}\sum_{\substack{b_1,\ldots,b_n\ge 0\\b_1+\ldots+b_n=2g}}P_{\alpha,d,g;\alpha_1,\ldots,\alpha_n}^{b_1,\ldots,b_n} u^{\alpha_1}_{b_1}\ldots u^{\alpha_n}_{b_n}.
$$
In particular, $\og_{\alpha,d}=\int g_{\alpha,d} dx$, expressed in terms of the $p$-variables, coincides with the definition given in \cite{Bur14}. The system of local functionals $\og_{\alpha,d}$, for $\alpha=1,\ldots,N$, $d=0,1,2,\ldots$, and the corresponding system of hamiltonian PDEs with respect to the standard Poisson bracket $\{\cdot,\cdot\}_{\eta\partial_x}$ is called the \emph{double ramification hierarchy}. The fact that the Hamiltonians $\og_{\alpha,d}$ mutually commute with respect to the standard bracket is proved in \cite{Bur14}. Finally, we add by hand $N$ more commuting hamiltonian densities $g_{\alpha,-1}:=\eta_{\alpha\mu} u^\mu$ for $\alpha=1,\ldots,N$. The corresponding local functionals $\og_{\alpha,-1}$ are Casimirs of the standard Poisson bracket.

\section{Recursion relations for the hamiltonian densities}

The results of this section are based on the following two splitting formulae from \cite{BSSZ12} for the intersection of a $\psi$-class with the double ramification cycle. Let  $I \sqcup J = \{ 1, \dots, n\}$ and let us denote by $\DR_{g_1}(a_I,-k_1, \dots, -k_p) \boxtimes \DR_{g_2}(a_J, k_1, \dots, k_p)$ the cycle in $\oM_{g_1+g_2+p-1,n}$ obtained by gluing the two double ramification cycles at the marked points labeled by $k_1,\ldots,k_p$.
\begin{theorem}[\cite{BSSZ12}] \label{BSSZ}
Let $a_1, \dots, a_n$ be a list of integers with vanishing sum. Assume that $a_s \not= 0$. Then we have
\begin{multline*}
a_s \psi_s \DR_g(a_1, \dots, a_n)=\\ 
=\sum_{I,J}
\sum_{p \geq 1}\sum_{g_1, g_2} 
\sum_{k_1, \dots, k_p}
 \frac{\rho}{2g-2+n} \frac{\prod_{i=1}^p k_i}{p!}
\DR_{g_1}(a_I,-k_1, \dots, -k_p) \boxtimes \DR_{g_2}(a_J, k_1, \dots, k_p).
\end{multline*}
Here the first sum is taken over all $I \sqcup J = \{ 1, \dots, n\}$ such that $\sum_{i \in I} a_i >0$; the third sum is over all non-negative genera $g_1$, $g_2$ satisfying $g_1+g_2+p-1=g$; the fourth sum is over $p$-uplets of positive integers with total sum $\sum_{i\in I}a_i=-\sum_{j\in J}a_j$. The number $\rho$ is defined by
$$
\rho = \left\{ 
\begin{array}{lcl}
 \hspace{0.4cm} 2g_2-2+|J|+p, & \mbox{if} & s \in I; \\
-(2g_1-2+|I|+p), & \mbox{if} & s \in J.
\end{array}
\right.
$$
\end{theorem}

\begin{theorem}[\cite{BSSZ12}] \label{BSSZ2}
Let $a_1, \dots, a_n$ be a list of integers with vanishing sum. Assume that $a_s \not= 0$ and $a_l=0$. Then we have
\begin{gather*}
a_s \psi_s \DR_g(a_1, \dots, a_n)=\sum_{I,J}
\sum_{p \geq 1}\sum_{g_1, g_2} 
\sum_{k_1, \dots, k_p}
\varepsilon \; \frac{\prod_{i=1}^p k_i}{p!}
\DR_{g_1}(a_I,-k_1, \dots, -k_p) \boxtimes \DR_{g_2}(a_J, k_1, \dots, k_p).
\end{gather*}
Here the first sum is taken over all $I \sqcup J = \{ 1, \dots, n\}$ such that $\sum_{i \in I} a_i >0$; the third sum is over all non-negative genera $g_1$, $g_2$ satisfying $g_1 + g_2 + p-1= g$; the fourth sum is over $p$-uplets of positive integers with total sum $\sum_{i \in I} a_i = -\sum_{j \in J} a_j$. The number $\varepsilon$ is defined by
$$
\varepsilon = 
\begin{cases}
 1, & \text{if $s \in I$ and $l \in J$};\\
-1, & \text{if $s \in J$ and $l \in I$};\\
0,  & \text{otherwise}.
\end{cases}
$$
\end{theorem}

\subsection{Dilaton recursion}
In this section we prove the most powerful of our recursion relations for the hamiltonian densities~\eqref{density}. It allows to reconstruct the full hierarchy of densities, starting from $\og_{1,1}$.

\begin{theorem}\label{dilaton}
We have the following recursion:
\begin{equation} \label{dilatoneq}
\del_x \left( (D-1) g_{\alpha,d+1}\right)= \sum_{k\ge 0} \left(\frac{\d g_{\alpha,d}}{\d u^\mu_k}\eta^{\mu\nu} \d^{k+1}_x\frac{\delta \og_{1,1}}{\delta u^\nu}\right), \hspace{1cm}\alpha=1,\ldots,N,\ d\geq -1,
\end{equation}
where $D:=\sum_{k\ge 0} (k+1) u^\alpha_k \frac{\del}{\del u^\alpha_k}$.
\end{theorem}
\begin{proof}
Given the polynomiality of the integral~\eqref{integral for the density}, it is sufficient to focus on the case, where $a_i>0, i=1,\ldots,n$.
Using Theorem \ref{BSSZ} we obtain
\begin{multline}\label{splitting}
(2g-2+n+1)\left(\sum_{i=1}^n a_i\right) P_{\alpha,d+1,g;\alpha_1,\ldots,\alpha_n}(a_1,\ldots,a_n) = \\
=\sum_{\substack{I\sqcup J=\{1,\ldots,n\}\\|J|\ge 1}}\sum_{\substack{g_1+g_2=g\\2g_2-1+|J|>0}}\sum_{k>0}(2g_2-2+|J|+1)k\cdot P_{\alpha,d,g_1;\alpha_I,\mu}(a_I,k) \eta^{\mu\nu} P_{g_2;\alpha_J,\nu}(a_J,-k),
\end{multline}
where $P_{g;\alpha_1,\ldots,\alpha_n}(a_1,\ldots,a_n)=\sum_{\substack{b_1,\ldots,b_n\ge 0\\b_1+\ldots+b_n=2g}} P^{b_1,\ldots,b_n}_{g;\alpha_1,\ldots,\alpha_n} a_1^{b_1}\ldots a_n^{b_n}=\int_{\DR_g\left(a_1,\ldots,a_n\right)}\lambda_g c_{g,n}(e_{\alpha_1}\otimes\ldots\otimes e_{\alpha_n})$. Notice that the right-hand side of the above formula is nonzero, only if $k=\sum_{j\in J} a_j$. Let us introduce an auxiliary functional
\begin{align*}
\og=&\sum_{\substack{g\ge 0,n\ge 2\\2g-2+n>0}}\frac{(-\eps^2)^g}{n!}\sum_{\substack{a_1,\ldots,a_n\in\mbZ\\\sum a_i=0}}P_{g;\alpha_1,\ldots,\alpha_n}(a_1,\ldots,a_n) p_{a_1}^{\alpha_1}\ldots p_{a_n}^{\alpha_n} = \\
= &\int\left(\sum_{\substack{g\ge 0,n\ge 2\\2g-2+n>0}}\frac{\eps^{2g}}{n!}\sum_{\substack{b_1,\ldots,b_n\ge 0\\b_1+\ldots+b_n=2g}}P^{b_1,\ldots,b_n}_{g;\alpha_1,\ldots,\alpha_n} u^{\alpha_1}_{b_1}\ldots u^{\alpha_n}_{b_n} \right) dx.
\end{align*}
Then equation~(\ref{splitting}) translates to
\begin{equation}\label{splitting2}
\del_x \left( (D-1) g_{\alpha,d+1}\right)=  \sum_{k\ge 0} \left(\frac{\d g_{\alpha,d}}{\d u^\mu_k}\eta^{\mu\nu} \d^{k+1}_x\frac{\delta}{\delta u^\nu} ((D-2) \og) \right).
\end{equation}
Using the formula for the push-forward of the class $\psi_1$ along the map $\pi\colon\oM_{g,n+1}\to \oM_{g,n}$, which forgets the first marked point, $\pi_*\psi_1= 2g-2 +n$, it's easy to prove the following version of the dilaton equation (compare also with \cite{FR10}):
\begin{gather}\label{eq:dilaton for og}
(D-2) \og = \og_{1,1},
\end{gather}
which, together with formula (\ref{splitting2}), proves the theorem for all $d \geq 0$. For the case $d=-1$, equation~\eqref{dilatoneq} gives $(D-1) g_{\alpha,0} = \frac{\delta \og_{1,1}}{\delta u^\alpha}$, which is again an immediate consequence of the dilaton equation above and the definition of $\og$.
\end{proof}

\begin{rem}
Notice how, in equation (\ref{dilatoneq}), the right-hand side can be written as
$$
\{g_{\alpha,d}, \og_{1,1}\}_{\eta \del_x} =\frac{\del g_{\alpha,d}}{\del t^1_1},
$$
where $t^\alpha_i$ is the time associated with the evolution along the hamiltonian flow generated by $\og_{\alpha,i}$. Since we know that $\{\og_{\alpha,d}, \og_{1,1}\}_{\eta \del_x}=0$, we are sure that the above expression is $\del_x$-exact and, hence, such is the right-hand side of equation (\ref{dilatoneq}), so that it makes sense to write
$$ (D-1) g_{\alpha,d+1}= \del_x^{-1}\frac{\del g_{\alpha,d}}{\del t^1_1}.$$
\end{rem}

\subsection{Topological recursion}
In this section we prove an equation for the hamiltonian densities that is reminiscent of the topological recursion relation in rational Gromov-Witten theory.
\begin{theorem}\label{trr}
We have
\begin{equation}\label{trreq}
\del_x \frac{\del g_{\alpha,d+1}}{\del u^\beta}= \sum_{k\ge 0} \left(\frac{\d g_{\alpha,d}}{\d u^\mu_k}\eta^{\mu\nu} \d^{k+1}_x\frac{\delta \og_{\beta,0}}{\delta u^\nu}\right), \hspace{1cm}1\le\alpha,\beta\le N,\ d\geq -1.
\end{equation}
\end{theorem}
\begin{proof}
The proof is completely analogous to the proof of Theorem \ref{dilaton}, but uses the second splitting formula, Theorem \ref{BSSZ2}. We leave the details to the reader.
\end{proof}
\begin{rem}
As for Theorem \ref{dilaton}, we can express Theorem \ref{trr} too in a more suggestive form
$$
\frac{\del g_{\alpha,d+1}}{\del u^\beta} =  \del_x^{-1} \frac{\del g_{\alpha,d}}{\del t^\beta_0}.
$$
As a special case, for $\beta=1$, given that $\del_{ t^1_0} = \del_x$, we get the following string-type equation:
$$ \frac{\del g_{\alpha,d+1}}{\del u^1} = g_{\alpha,d},$$
which was already proved in \cite{Bur14}.
\end{rem}

\subsection{Divisor equation}\label{subsection:divisor recursion}

In this section we derive another type of recursion that is based on the divisor equation in Gromov-Witten theory. We can recommend the papers~\cite{Get98} and~\cite{KM94} as a good introduction to Gromov-Witten theory.

Let $V$ be a smooth projective variety and suppose that $H^{odd}(V;\mbC)=0$. Denote by $E\subset H_2(V;\mbZ)$ the semigroup of effective classes and let $\mcN$ be the Novikov ring of~$V$. Let $e_1,\ldots,e_N$ be a basis in~$H^*(V;\mbC)$. We assume that the elements $e_i$ are homogeneous in the cohomology $H^*(V;\mbC)$. Consider the Gromov-Witten theory of $V$ and the corresponding cohomological field theory. Recall that the metric $\eta=(\eta_{\alpha\beta})$, $\eta_{\alpha\beta}=(e_\alpha,e_\beta)$, is induced by the Poincar\'e pairing in the cohomology $H^*(V;\mbC)$. Denote by 
$$
c_{g,n,\beta}(e_{\alpha_1}\otimes e_{\alpha_2}\otimes\ldots\otimes e_{\alpha_n})\in H^{even}(\oM_{g,n};\mbC)
$$
the classes of our cohomological field theory. Note that they depend now on a class $\beta\in E$. Let 
$$
c_{g,n}(e_{\alpha_1}\otimes\ldots\otimes e_{\alpha_n}):=\sum_{\beta\in E}q^\beta c_{g,n,\beta}(e_{\alpha_1}\otimes\ldots\otimes e_{\alpha_n})\in H^{even}(\oM_{g,n};\mbC)\otimes\mcN.
$$
Denote by $\left<\tau_{d_1}(e_{\alpha_1})\ldots\tau_{d_n}(e_{\alpha_n})\right>^{desc}_{g,\beta}$ the Gromov-Witten invariants of $V$. For any indices $1\le\alpha_1,\alpha_2,\alpha_3\le N$ let
$$
c_{\alpha_1\alpha_2\alpha_3}:=\sum_{\beta\in E}q^\beta\left<\tau_0(e_{\alpha_1})\tau_0(e_{\alpha_2})\tau_0(e_{\alpha_3})\right>^{desc}_{0,\beta}\qquad\text{and}\qquad c_{\alpha_1\alpha_2}^{\alpha_3}:=\eta^{\alpha_3\mu}c_{\alpha_1\alpha_2\mu}.
$$

We want to consider the double ramification hierarchy associated to our cohomological field theory. Note that our situation is slightly different from~\cite{Bur14} because of the presence of the Novikov ring. However, it is easy to see that all the constructions from~\cite{Bur14} still work. One should only keep in mind that now the Hamiltonians $\og_{\alpha,d}$ are elements of the space $\hLambda^{[0]}\otimes\mcN$ and the coefficients of the Hamiltonian operator $K$ belong to the space $\hcA\otimes\mcN$.

There is a natural operator $q\frac{\d}{\d q}\colon\mcN\to\mcN\otimes E$, $q\frac{\d}{\d q}(q^\beta):=q^\beta\otimes\beta$. The differential polynomial $g_{\alpha,d}$ is an element of the space $\hcA^{[0]}\otimes\mcN$, so $q\frac{\d}{\d q}g_{\alpha,d}$ is an element of $\hcA^{[0]}\otimes\mcN\otimes E$. Let $e_{\gamma_1},\ldots,e_{\gamma_r}$ be a basis in~$H^2(V;\mbC)$. Using the pairing $\<\cdot,\cdot\>\colon H^2(V;\mbC)\times E\to\mbC$, we define elements $\left<e_{\gamma_i},q\frac{\d}{\d q}g_{\alpha,d}\right>\in\hcA^{[0]}\otimes\mcN$. 

\begin{theorem}\label{divisor}
For any $i=1,\ldots,r$ and $d\ge -1$, we have
\begin{gather}\label{divisoreq}
\left<e_{\gamma_i},q\frac{\d}{\d q}g_{\alpha,d+1}\right>=\d_x^{-1}\frac{\d g_{\alpha,d}}{\d t^{\gamma_i}_0}-c^\mu_{\alpha\gamma_i}g_{\mu,d}.
\end{gather}
\end{theorem}
\begin{proof}
We begin by recalling the divisor equation in Gromov-Witten theory. Suppose that $2g-2+n>0$ and let $\pi_{n+1}\colon\oM_{g,n+1}\to\oM_{g,n}$ be the forgetful morphism that forgets the last marked point. The divisor equation says that (see e.g.~\cite{KM94})
\begin{gather}\label{eq:divisor equation in GW}
(\pi_{n+1})_*(c_{g,n+1,\beta}(e_{\alpha_1}\otimes\ldots\otimes e_{\alpha_n}\otimes e_{\gamma_i}))=\left<e_{\gamma_i},\beta\right>c_{g,n,\beta}(e_{\alpha_1}\otimes\ldots\otimes e_{\alpha_n}).
\end{gather}
Let us formulate the following simple lemma.
\begin{lemma}\label{lemma:simple divisor}
For any $i=1,\ldots,r$, we have 
\begin{gather}\label{eq:simple divisor}
\og_{\gamma_i,0}=\int\half c_{\gamma_i\mu\nu}u^\mu u^\nu dx+\left<e_{\gamma_i},q\frac{\d}{\d q}\og\right>.
\end{gather}
\end{lemma}
\begin{proof}
The proof is a simple consequence of the divisor equation~\eqref{eq:divisor equation in GW}.
\end{proof}
Let us take the variational derivative $\frac{\delta}{\delta u^\alpha}$ of the both sides of equation~\eqref{eq:simple divisor}. Since $\frac{\delta\og}{\delta u^\alpha}=g_{\alpha,0}$, we get exactly equation~\eqref{divisoreq} for $d=-1$.

Suppose $d\ge 0$. From the divisor equation~\eqref{eq:divisor equation in GW} it follows that
\begin{gather*}
\left<e_{\gamma_i},q\frac{\d}{\d q}  P_{\alpha,d+1,g;\alpha_1,\ldots,\alpha_n}(a_1,\ldots,a_n)
\right>=\int_{\DR_g(-\sum a_j,a_1,\ldots,a_n,0)}\pi_{n+2}^*(\psi_1^{d+1})\lambda_g c_{g,n+2}(e_\alpha\otimes \otimes_{j=1}^n e_{\alpha_j}\otimes e_{\gamma_i}).
\end{gather*}
We now express $\pi_{n+2}^*(\psi_1^{d+1})$ as
$$
\pi_{n+2}^*(\psi_1^{d+1})=\psi_1^{d+1}-\delta_0^{\{1,n+2\}}\cdot\pi_{n+2}^*(\psi_1^d).
$$
By Theorem~\ref{trr}, the first summand gives
\begin{align*}
&\sum_{\substack{g\ge 0\\n\ge 1}}\frac{(-\eps^2)^g}{n!}\sum_{a_1,\ldots,a_n\in\mbZ}\left(\int_{\DR_g(-\sum a_j,a_1,\ldots,a_n,0)}\psi_1^{d+1}\lambda_g c_{g,n+2}(e_\alpha\otimes\left(\otimes_{j=1}^n e_{\alpha_j}\right)\otimes e_{\gamma_i})\right)\prod_{j=1}^n p^{\alpha_j}_{a_j}e^{ix\sum a_j}=\\
=&\d_x^{-1}\frac{\d g_{\alpha,d}}{\d t^{\gamma_i}_0},
\end{align*}
which coincides with the first term on the right-hand side of~(\ref{divisoreq}). For the second summand, we have (see~\cite{BSSZ12}) 
$$
\delta_0^{\{1,n+2\}}\cdot \DR_g\left(-\sum a_j,a_1,\ldots,a_n,0\right)=\DR_0\left(-\sum a_j,0,\sum a_j\right)\boxtimes \DR_g\left(a_1,\ldots,a_n,-\sum a_j\right),
$$
from which we get
\begin{multline*}
\int_{\DR_g(-\sum a_j,a_1,\ldots,a_n,0)}\delta_0^{\{1,n+2\}}\pi_{n+2}^*(\psi_1^d)\lambda_g c_{g,n+2}(e_\alpha\otimes(\otimes_{j=1}^n e_{\alpha_j})\otimes e_{\gamma_i})=\\
=c^\mu_{\alpha\gamma_i}\int_{\DR_g(-\sum a_j,a_1,\ldots,a_n)}\psi_1^d\lambda_g c_{g,n+1}(e_\mu\otimes(\otimes_{j=1}^n e_{\alpha_j}))=c^\mu_{\alpha\gamma_i}P_{\mu,d,g;\alpha_1,\ldots,\alpha_n}(a_1,\ldots,a_n).
\end{multline*}
This corresponds to the second term on the right-hand side of equation~(\ref{divisoreq}). The theorem is proved.
\end{proof}

\section{Examples and applications}

The recursion formulae we proved in the previous section are computationally very efficient and allow us to produce a number of calculations in concrete examples. Beside the case of the complex projective line, which we leave for the next section, we focus here on computing various relevant quantities for one-dimensional cohomological field theories and for some of Witten's $r$-spin theories, which are ultimately sufficient to determine the full double ramification hierarchy in these cases.

\subsection{KdV hierarchy}

The simplest cohomological field theory, $V=\langle e_1\rangle$, $\eta_{1,1}=1$, $c_{g,n}(e_1^{\otimes n})=1$ for all stable $(g, n)$, corresponding to the Gromov-Witten theory of a point, gives as the double ramification hierarchy the Korteweg-de Vries hierarchy: the equivalence conjecture between the Dubrovin-Zhang hierarchy and the double ramification hierarchy holds in this case with the trivial Miura transformation, as proved in~\cite{Bur14}. However, our recursion formulae give a specific choice of hamiltonian densities $g_{d}$ which was previously unknown and is non-standard. In fact, such hamiltonian densities do not satisfy the tau-symmetry property $\{h_{p-1},\oh_q\}_{\partial_x} = \{h_{q-1},\oh_p\}_{\partial_x}$. However, the standard tau-symmetric hamiltonian densities $h_p$ for the Dubrovin-Zhang hierarchy can be recovered as $h_p=\frac{\delta \og_{p+1}}{\delta u}$. We will see in the examples below that this is a quite general fact (see also remark \ref{remarktau}).

\begin{example}
For the trivial cohomological field theory we have
$$
\og_1=\int \left( \frac{u^3}{6}+\frac{\eps^2}{24} u u_2\right) dx,
$$
which determines the following hamiltonian densities for the double ramification hierarchy:
\begin{equation*}
\begin{split}
&g_{-1}=u,\\
&g_0=\frac{u^2}{2}+\frac{\eps ^2}{24}u_2,  \\
&g_1=\frac{u^3}{6}+\frac{ \eps ^2}{24} u u_2+\frac{ \eps ^4}{1152}u_4,\\
&g_2=\frac{u^4}{24}+\frac{\eps^2}{48} u^2 u_2+\left(\frac{7 \left(u_2\right){}^2}{5760}+\frac{u u_4}{1152}\right) \eps ^4+\frac{\eps^6}{82944}u_6,\\
&g_3=\frac{u^5}{120}+\frac{\eps^2}{144}u^3 u_2+\left(\frac{7 u \left(u_2\right){}^2}{5760}+\frac{u^2 u_4}{2304}\right) \eps^4+\left(\frac{\left(u_3\right){}^2}{362880}+\frac{u_2 u_4}{15360}+\frac{u u_6}{82944}\right) \eps ^6+\frac{\eps^8}{7962624}u_8,\\
\end{split}
\end{equation*}
and so on. These hamiltonian densities, as remarked above, integrate to the usual KdV local functionals $\og_p=\oh_p$ for the Dubrovin-Zhang hierarchy. In particular, if 
$$
\chi = \lambda^{\frac{1}{2}} + \sum_{m=1}^\infty \frac{\chi_m}{\lambda^{m/2}}
$$
is a solution to the Riccati equation 
$$
\frac{\eps}{\sqrt{2}} \chi' - \chi^2 = u - \lambda,
$$
then
$$
\og_p=-\frac{2^{p+2}}{(2p+1)!!}\int\chi_{2p+3}dx,
$$
and the hamiltonian densities $h_k:= \frac{\delta \og_{k+1}}{\delta u} =-\frac{2^{k+3}}{(2k+3)!!}\frac{\delta}{\delta u}\int\chi_{2k+5}dx$, for $k\geq -1$, are the tau-symmetric densities of the Dubrovin-Zhang hierarchy, with $h_{-1}=g_{-1}=u$ (see also \cite{DZ05}).
\end{example}

\subsection{Intermediate Long Wave hierarchy}

In \cite{Bur14} one finds a proof of the Miura equivalence between the Dubrovin-Zhang and the double ramification hierarchy for the case of the cohomological field theory (depending on a parameter $\ell$) consisting of the full Hodge class $c_{g,n}(e_1^{\otimes n})=1+\ell \lambda_1+\ldots+\ell^g \lambda_g$, with $V=\langle e_1\rangle$, $\eta_{1,1}=1$.

\begin{example}
For the cohomological field theory given by the full Hodge class we have
$$
\og_1=\int \left( \frac{u^3}{6}+ \sum_{g\geq 1} \eps^{2g} \ell^{g-1} \frac{|B_{2g}|}{2 (2g)!} u u_{2g} \right)dx,
$$
where $B_{2g}$ are Bernoulli numbers: $B_0=1, B_2=\frac{1}{6}, B_4=-\frac{1}{30},\ldots$. This Hamiltonian, by our recursion, determines the full double ramification hierarchy. In~\cite{Bur13} this hierarchy was called the deformed KdV hierarchy. The Miura transformation
$$
u\mapsto\widetilde{u}= u+\sum_{g\geq 1}\frac{2^{2g-1}-1}{2^{2g-1}}\frac{|B_{2g}|}{(2g)!}\eps^{2g} \ell^g u_{2g}
$$
transforms this hierarchy to the Dubrovin-Zhang hierarchy. In particular, the standard hamiltonian operator $\partial_x$ is transformed to the hamiltonian operator 
$$
K=\partial_x + \sum_{g\geq 1}\eps^{2g}\ell^g\frac{(2g-1) |B_{2g}|}{(2g)!} \partial_x^{2g+1}.
$$
In~\cite{Bur13} it is explained how the deformed KdV hierarchy is related to the hierarchy of the conserved quantities of the Intermediate Long Wave (ILW) equation (see e.g. \cite{SAK79}):
$$
w_\tau + 2 w w_x + T(w_{xx})=0, \hspace{0.5cm} T(f):=\mathrm{p.v.}\int_{-\infty}^{+\infty} \frac{1}{2\delta} \left(\mathrm{sgn}(x-\xi) - \coth \frac{\pi (x-\xi)}{2 \delta}\right) f(\xi) d\xi.
$$
The ILW equation can be transformed to the first equation of the deformed KdV hierarchy by setting $w=\frac{\sqrt\ell}{\eps} u$, $\tau=-\frac{1}{2}\frac{\eps}{\sqrt\ell} t_1$, $\delta= \frac{\eps\sqrt\ell}{2}$ (indeed $T(f)=\sum_{n\geq 1} \delta^{2n-1} 2^{2n} \frac{|B_{2n}|}{(2n)!}\partial_x^{2n-1} f$, see~\cite{Bur13} for further details). This means that our recursion formula gives a way, alternative to~\cite{SAK79}, to determine the symmetries of the ILW equation.
\end{example}

\subsection{Higher spin / Gelfand-Dickey hierarchies}

Recall that, for every $r\geq 2$ and an $(r-1)$-dimensional vector space $V$ with a basis $e_1,\ldots,e_{r-1}$, Witten's $r$-spin classes 
$$
W_g(e_{a_1+1},\ldots,e_{a_n+1})=W_g(a_1,\ldots,a_n) \in H^{even}(\oM_{g,n};\Q)
$$
are cohomology classes of degree 
$$
\deg W_g(a_1,\ldots,a_n)=2\left(\frac{(r-2)(g-1)+\sum_{i=1}^n a_i}{r}\right),
$$
when $a_i\in\{0,\ldots,r-2\}$ are such that the expression in the brackets on the right-hand side is an integer, and vanish otherwise. They form a cohomological field theory and were introduced by Witten~\cite{Wi93} in genus $0$ and then extended to higher genus by Polishchuck and Vaintrob \cite{PV00} (see also \cite{Ch06}). As proved in \cite{PPZ13}, this cohomological field theory is completely determined, thanks to semisimplicity, by the initial conditions: 
\begin{align*}
W_0(a_1,a_2,a_3)=&
\begin{cases}
1,&\text{if $a_1+a_2+a_3=r-2$};\\
0,&\text{otherwise};
\end{cases}\\
W_0(1,1,r-2,r-2)=&\frac{1}{r}[\text{point}]\in H^2(\oM_{0,4};\mathbb{Q}).
\end{align*}
In particular, the metric $\eta$ takes the form $\eta_{\alpha\beta}=\delta_{\alpha+\beta,r}$.\\

Using our recursion formulae together with the selection rules from the degree formula for the $r$-spin classes, it is possible to completely determine the Hamiltonian $\og_{1,1}$ and, hence, the full hierarchy. In particular, from dimension counting and the definition (\ref{density}), we obtain that $g_{1,1}$ is a homogeneous polynomial of degree $2r+2$ with respect to the following grading:
$$
|u^{a+1}_k|=r-a, \hspace{0.3cm} a=0,\ldots,r-2,\ k=0,1,\ldots; \hspace{1cm}|\eps|=1.
$$
This gives a finite number of summands involving intersection numbers only up to genus $r$ for $g_{1,1}$, and up to genus $r-1$ for $\og_{1,1}$, as the top genus term is $\partial_x$-exact. At this point, one can start applying the recursion of Theorem \ref{dilaton} starting from $g_{\alpha,-1}=\eta_{\alpha\mu} u^\mu$ and impose, at each step, that the new Hamiltonians thus obtained still commute with all the others. In the cases $r=3,4$ this determines all the coefficients of the Hamiltonian $\og_{1,1}$ up to a rescaling of the form $\eps\mapsto\theta\eps$, $\theta\in\mathbb{Q}$. This ambiguity can be fixed by the following computation. For any cohomological field theory $c_{g,n}\colon V^{\otimes n}\to H^{even}(\oM_{g,n};\mbC)$ we have
\begin{align*}
\Coef_{\eps^2 u^1_{xx}}\left(\frac{\delta\og_{1,1}}{\delta u^1}\right)=&\Coef_{a^2}\int_{\DR_1(0,-a,a)}\psi_1\lambda_1c_{1,3}(e_1^{\otimes 3})=2\Coef_{a^2}\int_{\DR_1(-a,a)}\lambda_1c_{1,2}(e_1^{\otimes 2})\stackrel{\text{by \eqref{eq:Hain's formula}}}{=}\\
=&2\int_{\oM_{1,2}}\delta_0^{\{1,2\}}\lambda_1c_{1,2}(e_1^{\otimes 2})=2\int_{\oM_{1,1}}\lambda_1c_{1,1}(e_1)=\frac{1}{12}\int_{\oM_{1,1}}\delta_0^{\mathrm{ns}}c_{1,1}(e_1)=\frac{\dim V}{12}.
\end{align*}
Here $\delta_0^{\mathrm{ns}}$ represents the divisor whose generic point is a nodal curve with a non-separating node and we also used that on $\oM_{1,1}$ we have $\lambda_1=\psi_1^{\dagger}=\frac{1}{24}\delta_0^{\mathrm{ns}}$.
\begin{example}
For Witten's $3$-spin cohomological field theory we have
\begin{equation*}
\og_{1,1}=\int\left[\left(\frac{1}{2} \left(u^1\right)^2 u^2+\frac{\left(u^2\right)^4}{36}\right)+\left(-\frac{1}{12} \left(u_1^1\right){}^2-\frac{1}{24} u^2 \left(u_1^2\right){}^2\right) \eps ^2+\frac{1}{432} \left(u_2^2\right){}^2
   \eps ^4\right] dx,
\end{equation*}
which determines the following hamiltonian densities for the double ramification hierarchy:
\begin{equation*}
\left\{
\begin{array}{l}
g_{1,-1}=u^2,\\ 
g_{2,-1}=u^1; 
\end{array}
\right. \hspace{1.5cm}
\left\{
\begin{array}{l}
g_{1,0}=u^1 u^2+\frac{1}{12} u_2^1 \eps ^2,\\
g_{2,0}=\frac{\left(u^1\right)^2}{2}+\frac{\left(u^2\right)^3}{18}+\left(\frac{1}{72} \left(u_1^2\right){}^2+\frac{1}{36} u^2 u_2^2\right) \eps ^2+\frac{1}{864} u_4^2 \eps ^4;\\
\end{array}
\right.
\end{equation*}

\begin{equation*}
\left\{
\begin{array}{l}
\scriptstyle{g_{1,1}=\frac{1}{2} \left(u^1\right)^2 u^2+\frac{\left(u^2\right)^4}{36}+\left(\frac{1}{72} u^2 \left(u_1^2\right){}^2+\frac{1}{12} u^1 u_2^1+\frac{1}{36} \left(u^2\right)^2 u_2^2\right) \eps ^2+\left(\frac{7
   \left(u_2^2\right){}^2}{2160}+\frac{7 u_1^2 u_3^2}{2160}+\frac{1}{432} u^2 u_4^2\right) \eps ^4+\frac{u_6^2 \eps ^6}{15552}},\\
\scriptstyle{g_{2,1}=\frac{\left(u^1\right)^3}{6}+\frac{1}{18} u^1 \left(u^2\right)^3+\left(\frac{1}{72} u^1 \left(u_1^2\right){}^2+\frac{1}{72} \left(u^2\right)^2 u_2^1+\frac{1}{36} u^1 u^2 u_2^2\right) \eps
   ^2+\left(\frac{1}{432} u_2^1 u_2^2+\frac{u_1^2 u_3^1}{1080}+\frac{1}{864} u^2 u_4^1+\frac{1}{864} u^1 u_4^2\right) \eps ^4+\frac{u_6^1 \eps ^6}{31104}};\\
\end{array}
\right.
\end{equation*}

\tiny
\begin{equation*}
\left\{ \begin{array}{l}
\begin{split} g_{1,2}&=\left(\frac{1}{6} \left(u^1\right)^3 u^2+\frac{1}{36} u^1 \left(u^2\right)^4\right)+\left(\frac{1}{72} u^1 u^2 \left(u_1^2\right){}^2+\frac{1}{24} \left(u^1\right)^2 u_2^1+\frac{1}{108} \left(u^2\right)^3 u_2^1+\frac{1}{36} u^1 \left(u^2\right)^2 u_2^2\right) \eps ^2\\
&+\left(\frac{7 \left(u_1^2\right){}^2 u_2^1}{4320}+\frac{1}{180} u^2 u_2^1 u_2^2+\frac{7 u^1 \left(u_2^2\right){}^2}{2160}+\frac{u^2 u_1^2
   u_3^1}{1080}+\frac{7 u^1 u_1^2 u_3^2}{2160}+\frac{1}{864} \left(u^2\right)^2 u_4^1+\frac{1}{432} u^1 u^2 u_4^2\right) \eps ^4\\
   &+\left(\frac{u_3^1 u_3^2}{7776}+\frac{7 u_2^2 u_4^1}{25920}+\frac{u_2^1
   u_4^2}{4032}+\frac{u_1^2 u_5^1}{12960}+\frac{u^2 u_6^1}{15552}+\frac{u^1 u_6^2}{15552}\right) \eps ^6+\frac{u_8^1 \eps ^8}{746496},
\end{split}\\
\begin{split} g_{2,2}&=\frac{\left(u^1\right)^4}{24}+\frac{1}{36} \left(u^1\right)^2 \left(u^2\right)^3+\frac{\left(u^2\right)^6}{1620}+\left(\frac{1}{144} \left(u^1\right)^2 \left(u_1^2\right){}^2+\frac{\left(u^2\right)^3
   \left(u_1^2\right){}^2}{1296}+\frac{1}{72} u^1 \left(u^2\right)^2 u_2^1+\frac{1}{72} \left(u^1\right)^2 u^2 u_2^2+\frac{1}{648} \left(u^2\right)^4 u_2^2\right) \eps ^2\\
   &+\left(\frac{7 \left(u_1^2\right){}^4}{51840}+\frac{7 u^2 \left(u_2^1\right){}^2}{4320}+\frac{7 u^2 \left(u_1^2\right){}^2 u_2^2}{12960}+\frac{u^1 u_2^1 u_2^2}{432} +\frac{7 \left(u^2\right)^2 \left(u_2^2\right){}^2}{6480}+\frac{u^1
   u_1^2 u_3^1}{1080}+\frac{7 \left(u^2\right)^2 u_1^2 u_3^2}{12960}+\frac{ u^1 u^2 u_4^1}{864}+\frac{\left(u^1\right)^2 u_4^2}{1728}+ \frac{\left(u^2\right)^3 u_4^2}{3888}\right) \eps ^4\\
   & +\left(\frac{11
   \left(u_2^2\right){}^3}{116640}+\frac{\left(u_3^1\right){}^2}{34020}+\frac{13 u_1^2 u_2^2 u_3^2}{77760}+\frac{41 u^2 \left(u_3^2\right){}^2}{466560}+\frac{11 u_2^1 u_4^1}{72576}+\frac{17 \left(u_1^2\right){}^2
   u_4^2}{311040}+\frac{u^2 u_2^2 u_4^2}{4320}+\frac{u^2 u_1^2 u_5^2}{17280}+\frac{u^1 u_6^1}{31104}+\frac{\left(u^2\right)^2 u_6^2}{46656}\right) \eps ^6\\
   &+\left(\frac{47 \left(u_4^2\right){}^2}{5598720}+\frac{61 u_3^2
   u_5^2}{5598720}+\frac{11 u_2^2 u_6^2}{1399680}+\frac{11 u_1^2 u_7^2}{5598720}+\frac{u^2 u_8^2}{1119744}\right) \eps ^8+\frac{u_{10}^2 \eps ^{10}}{67184640};\\
\end{split}
\end{array}\right.
\end{equation*}
\normalsize
and so on (we have explicit formulae up to $g_{\alpha,7}$). We remark that, by taking the covariant derivative with respect to $u^1$ of the local functionals associated to the densities we computed here, $h_{\alpha,p}:=\frac{\delta \og_{\alpha,p+1}}{\delta u^1}$, one finds precisely the tau-symmetric hamiltonian densities for the Gelfand-Dickey hierarchy \cite{GD76} associated with the $A_2$ Coxeter group \cite{Dub96} (the so-called Boussinesq hierarchy) which coincide with the densities in the normal coordinates for the Dubrovin-Zhang hierarchy for the $3$-spin classes. This gives a strong evidence for the Miura equivalence of the Dubrovin-Zhang and the double ramification hierarchy, with the trivial Miura tranformation, for the $3$-spin classes.
\end{example}

\begin{example}
For Witten's $4$-spin cohomological field theory we have
\scriptsize
\begin{equation*}
\begin{split}
\og_{1,1}=\int &\left[\frac{ u^1 \left(u^2\right)^2}{2}+\frac{\left(u^1\right)^2 u^3}{2} +\frac{\left(u^2\right)^2 \left(u^3\right)^2}{8} +\frac{\left(u^3\right)^5}{320}+\left(-\frac{\left(u_1^1\right){}^2}{8} -\frac{u^3
   \left(u_1^2\right){}^2}{16} -\frac{u^3 u_1^1 u_1^3}{32} +\frac{3}{64} \left(u^2\right)^2 u_2^3+\frac{1}{192} \left(u^3\right)^3 u_2^3\right) \eps ^2 \right. \\
   & \left. +\left(\frac{1}{160} \left(u_2^2\right){}^2+\frac{3}{640} u_2^1
   u_2^3+\frac{5 \left(u^3\right)^2 u_4^3}{4096}\right) \eps ^4-\frac{\left(u_3^3\right){}^2 \eps ^6}{8192}\right] dx,
\end{split}
 \end{equation*}
 \normalsize
which determines the following hamiltonian densities for the double ramification hierarchy:
\scriptsize
\begin{equation*}
\left\{
\begin{array}{l}
g_{1,-1}=u^3,\\ 
g_{2,-1}=u^2,\\
g_{3,-1}=u^1;
\end{array}
\right. \hspace{1.5cm}
\left\{
\begin{array}{l}
 g_{1,0}= \left(\frac{\left(u^2\right)^2}{2}+u^1 u^3\right)+\left(\frac{1}{96} \left(u_1^3\right){}^2+\frac{u_2^1}{8}+\frac{1}{96} u^3 u_2^3\right) \eps ^2+\frac{3 u_4^3 \eps ^4}{2560},\\
 g_{2,0}=\left(u^1 u^2+\frac{1}{8} u^2 \left(u^3\right)^2\right)+\left(\frac{1}{24} u_1^2 u_1^3+\frac{1}{24} u^3 u_2^2+\frac{1}{32} u^2 u_2^3\right) \eps ^2+\frac{1}{320} u_4^2 \eps ^4,\\
\begin{split} g_{3,0}=&\left(\frac{\left(u^1\right)^2}{2}+\frac{1}{8} \left(u^2\right)^2 u^3+\frac{\left(u^3\right)^4}{192}\right)+\left(\frac{1}{96} \left(u_1^2\right){}^2+\frac{1}{128} u^3 \left(u_1^3\right){}^2+\frac{1}{96} u^3
   u_2^1+\frac{1}{32} u^2 u_2^2 \right.\\
   & \left.+\frac{1}{128} \left(u^3\right)^2 u_2^3\right) \eps ^2+\left(\frac{3 \left(u_2^3\right){}^2}{2048}+\frac{1}{512} u_1^3 u_3^3+\frac{3 u_4^1}{2560}+\frac{u^3 u_4^3}{1024}\right) \eps
   ^4+\frac{u_6^3 \eps ^6}{24576};\end{split}\\
\end{array}
\right.
\end{equation*}
\begin{equation*}
\left\{
\begin{array}{l}
\begin{split}
g_{1,1}=&\left(\frac{1}{2} u^1 \left(u^2\right)^2+\frac{1}{2} \left(u^1\right)^2 u^3+\frac{1}{8} \left(u^2\right)^2 \left(u^3\right)^2+\frac{\left(u^3\right)^5}{320}\right)+\left(\frac{1}{96} u^3 \left(u_1^2\right){}^2+\frac{1}{24}
   u^2 u_1^2 u_1^3+\frac{1}{96} u^1 \left(u_1^3\right){}^2+\frac{1}{128} \left(u^3\right)^2 \left(u_1^3\right){}^2 \right. \\
   & \left.+\frac{1}{8} u^1 u_2^1+\frac{1}{96} \left(u^3\right)^2 u_2^1+\frac{7}{96} u^2 u^3 u_2^2+\frac{1}{32}
   \left(u^2\right)^2 u_2^3+\frac{1}{96} u^1 u^3 u_2^3+\frac{1}{128} \left(u^3\right)^3 u_2^3\right) \eps ^2+\left(\frac{3}{512} \left(u_2^2\right){}^2+\frac{1}{256} \left(u_1^3\right){}^2 u_2^3\right. \\
   & \left.+\frac{1}{320} u_2^1
   u_2^3+\frac{9 u^3 \left(u_2^3\right){}^2}{2048}+\frac{1}{480} u_1^3 u_3^1+\frac{3}{640} u_1^2 u_3^2+\frac{23 u^3 u_1^3 u_3^3}{4608}+\frac{19 u^3 u_4^1}{7680}+\frac{13 u^2 u_4^2}{2560}+\frac{3 u^1 u_4^3}{2560}+\frac{7
   \left(u^3\right)^2 u_4^3}{4608}\right) \eps ^4 \\
   &+\left(\frac{27 \left(u_3^3\right){}^2}{57344}+\frac{93 u_2^3 u_4^3}{114688}+\frac{101 u_1^3 u_5^3}{286720}+\frac{3 u_6^1}{20480}+\frac{11 u^3 u_6^3}{81920}\right)
   \eps ^6+\frac{59 u_8^3 \eps ^8}{13107200},\end{split}\\
\begin{split}g_{2,1}&=\left(\frac{1}{2} \left(u^1\right)^2 u^2+\frac{1}{12} \left(u^2\right)^3 u^3+\frac{1}{8} u^1 u^2 \left(u^3\right)^2+\frac{1}{128} u^2 \left(u^3\right)^4\right)+\left(\frac{1}{96} u^2 \left(u_1^2\right){}^2+\frac{1}{24} u^1
   u_1^2 u_1^3+\frac{1}{192} \left(u^3\right)^2 u_1^2 u_1^3+\frac{1}{96} u^2 u^3 \left(u_1^3\right){}^2\right. \\
   & \left.+\frac{1}{24} u^2 u^3 u_2^1+\frac{1}{32} \left(u^2\right)^2 u_2^2+\frac{1}{24} u^1 u^3 u_2^2+\frac{1}{192}
   \left(u^3\right)^3 u_2^2+\frac{1}{32} u^1 u^2 u_2^3+\frac{11}{768} u^2 \left(u^3\right)^2 u_2^3\right) \eps ^2+\left(\frac{1}{480} \left(u_1^3\right){}^2 u_2^2+\frac{1}{160} u_2^1 u_2^2\right. \\
   & \left.+\frac{11 u_1^2 u_1^3
   u_2^3}{3840}+\frac{23 u^3 u_2^2 u_2^3}{3840}+\frac{29 u^2 \left(u_2^3\right){}^2}{10240}+\frac{1}{480} u_1^2 u_3^1+\frac{1}{320} u^3 u_1^3 u_3^2+\frac{13 u^3 u_1^2 u_3^3}{5760}+\frac{1}{320} u^2 u_1^3
   u_3^3+\frac{1}{320} u^2 u_4^1+\frac{1}{320} u^1 u_4^2\right. \\
   & \left.+\frac{29 \left(u^3\right)^2 u_4^2}{23040}+\frac{1}{480} u^2 u^3 u_4^3\right) \eps ^4+\left(\frac{3 u_3^2 u_3^3}{4480}+\frac{47 u_2^3 u_4^2}{71680}+\frac{3 u_2^2
   u_4^3}{5120}+\frac{u_1^3 u_5^2}{3584}+\frac{u_1^2 u_5^3}{5120}+\frac{u^3 u_6^2}{7680}+\frac{u^2 u_6^3}{10240}\right) \eps ^6+\frac{u_8^2 \eps ^8}{204800},\end{split}\\
 \begin{split}g_{3,1}&=\left(\frac{\left(u^1\right)^3}{6}+\frac{\left(u^2\right)^4}{96}+\frac{1}{8} u^1 \left(u^2\right)^2 u^3+\frac{1}{96} \left(u^2\right)^2 \left(u^3\right)^3+\frac{1}{192} u^1 \left(u^3\right)^4\right)+\left(\frac{1}{96} u^1
   \left(u_1^2\right){}^2+\frac{1}{96} u^2 u^3 u_1^2 u_1^3+\frac{1}{768} \left(u^2\right)^2 \left(u_1^3\right){}^2\right. \\
   & \left.+\frac{1}{128} u^1 u^3 \left(u_1^3\right){}^2+\frac{\left(u^3\right)^3
   \left(u_1^3\right){}^2}{4608}+\frac{1}{64} \left(u^2\right)^2 u_2^1+\frac{1}{96} u^1 u^3 u_2^1+\frac{1}{384} \left(u^3\right)^3 u_2^1+\frac{1}{32} u^1 u^2 u_2^2+\frac{1}{96} u^2 \left(u^3\right)^2 u_2^2+\frac{7}{768}
   \left(u^2\right)^2 u^3 u_2^3\right. \\
   & \left.+\frac{1}{128} u^1 \left(u^3\right)^2 u_2^3+\frac{\left(u^3\right)^4 u_2^3}{4608}\right) \eps ^2+\left(\frac{\left(u_1^3\right){}^4}{40960}+\frac{7 \left(u_1^3\right){}^2
   u_2^1}{9216}+\frac{11 \left(u_2^1\right){}^2}{7680}+\frac{23 u_1^2 u_1^3 u_2^2}{11520}+\frac{161 u^3 \left(u_2^2\right){}^2}{92160}+\frac{13 \left(u_1^2\right){}^2 u_2^3}{23040}\right. \\
   & \left.+\frac{13 u^3 \left(u_1^3\right){}^2
   u_2^3}{20480}+\frac{19 u^3 u_2^1 u_2^3}{9216}+\frac{31 u^2 u_2^2 u_2^3}{7680}+\frac{3 u^1 \left(u_2^3\right){}^2}{2048}+\frac{19 \left(u^3\right)^2 \left(u_2^3\right){}^2}{40960}+\frac{u^3 u_1^3 u_3^1}{1152}+\frac{17
   u^3 u_1^2 u_3^2}{23040}+\frac{7 u^2 u_1^3 u_3^2}{3840}+\frac{7 u^2 u_1^2 u_3^3}{3840}\right. \\
   & \left.+\frac{1}{512} u^1 u_1^3 u_3^3+\frac{7 \left(u^3\right)^2 u_1^3 u_3^3}{15360}+\frac{3 u^1 u_4^1}{2560}+\frac{5 \left(u^3\right)^2
   u_4^1}{9216}+\frac{49 u^2 u^3 u_4^2}{30720}+\frac{13 \left(u^2\right)^2 u_4^3}{20480}+\frac{u^1 u^3 u_4^3}{1024}+\frac{13 \left(u^3\right)^3 u_4^3}{122880}\right) \eps ^4+\left(\frac{7
   \left(u_2^3\right){}^3}{61440}\right. \\
   & \left.+\frac{33 \left(u_3^2\right){}^2}{286720}+\frac{u_1^3 u_2^3 u_3^3}{2560}+\frac{u_3^1 u_3^3}{7680}+\frac{21 u^3 \left(u_3^3\right){}^2}{163840}+\frac{u_2^3 u_4^1}{5120}+\frac{77 u_2^2
   u_4^2}{245760}+\frac{33 \left(u_1^3\right){}^2 u_4^3}{327680}+\frac{11 u_2^1 u_4^3}{81920}+\frac{19 u^3 u_2^3 u_4^3}{81920}+\frac{3 u_1^3 u_5^1}{35840}+\frac{11 u_1^2 u_5^2}{122880}\right. \\
   & \left.+\frac{3 u^3 u_1^3
   u_5^3}{32768}+\frac{13 u^3 u_6^1}{245760}+\frac{19 u^2 u_6^2}{245760}+\frac{u^1 u_6^3}{24576}+\frac{19 \left(u^3\right)^2 u_6^3}{983040}\right) \eps ^6+\left(\frac{11 \left(u_4^3\right){}^2}{655360}+\frac{7 u_3^3
   u_5^3}{262144}+\frac{379 u_2^3 u_6^3}{23592960}+\frac{61 u_1^3 u_7^3}{11796480}+\frac{77 u_8^1}{39321600}\right. \\
   & \left.+\frac{37 u^3 u_8^3}{23592960}\right) \eps ^8+\frac{u_{10}^3 \eps ^{10}}{20971520};\end{split}\\
\end{array}
\right. 
\end{equation*}
\normalsize
and so on (we have explicit formulae up to $g_{\alpha,4}$). We want to remark that in this case, as opposed to the $3$-spin case, if one defines the tau-symmetric densities $h_{\alpha,p}:=\frac{\delta \og_{\alpha,p+1}}{\delta u^1}$, then the coordinates $u^\alpha$ are not normal in Dubrovin and Zhang's sense anymore, i.e the starting hamiltonian densities $h_{\alpha,-1}$ take the non-standard form
\begin{equation*}
\left\{
\begin{array}{l}
h_{1,-1}=u^3,\\ 
h_{2,-1}=u^2, \\
h_{3,-1}=u^1+\frac{1}{96} u_2^3 \eps ^2.
\end{array}
\right. 
\end{equation*}
One can then perform a Miura transformation $w^\alpha = \eta^{\alpha \mu} h_{\mu,-1}$ to pass to the appropriate normal coordinates and in these coordinates the Poisson structure changes to the one associated to the hamiltonian operator
\begin{equation*}\left(
\begin{array}{c c c}
\frac{1}{48}\eps^2 \partial_x^3 & 0 & \partial_x \\
0 & \partial_x & 0 \\
\partial_x & 0 & 0
\end{array}
\right).
\end{equation*}
When expressed in these new coordinates, both the Poisson structure and the hamiltonian densities~$h_{\alpha,p}$ coincide with the ones for the Gelfand-Dickey \cite{GD76} hierarchy associated with the $A_3$ Coxeter group \cite{Dub96}, or the dispersive Poisson structure and the hamiltonian densities in the normal coordinates for the Dubrovin-Zhang hierarchy for the $4$-spin classes. As above, this gives a strong evidence for the Miura equivalence of the Dubrovin-Zhang and the double ramification hierarchy for the $4$-spin classes, with respect to the Miura transformation $w^\alpha = \eta^{\alpha \mu} h_{\mu,-1}$.
\end{example}

\begin{rem} \label{remarktau}
The same technique of this section can actually be applied to any polynomial Frobenius manifold (in particular, to all Frobenius manifolds associated to the Coxeter groups) and one obtains similar conjectures about an explicit form of a Miura transformation connecting the double ramification hierarchy to the Dubrovin-Zhang hierarchy. We plan to address the problem of understanding a connection between the tau-symmetric hamiltonian densities for the double ramification hierarchy, the normal coordinates and an equivalence to the Dubrovin-Zhang hierarchy in a forthcoming paper.
\end{rem}


\section{Divisor equation for the DR hierarchies}\label{section:divisor equation}

In this section we derive a certain equation for the string solution of the double ramification hierarchy. This equation is very similar to the divisor equation in Gromov-Witten theory. In Section~\ref{subsection:property of the string solution} we derive a useful property of the string solution. In~Sections~\ref{subsection:divisor equation for the Hamiltonians} and~\ref{subsection:divisor equation for the string solution} we consider the cohomological field theory accociated with the Gromov-Witten theory of some target variety~$V$. In Section~\ref{subsection:divisor equation for the Hamiltonians} we prove a divisor equation for the Hamiltonians of the double ramification hierarchy. Section~\ref{subsection:divisor equation for the string solution} is devoted to the proof of a divisor equation for the string solution.  

\subsection{Property of the string solution}\label{subsection:property of the string solution}

Consider an arbitrary cohomological field theory and the corresponding double ramification hierarchy. Recall that the string solution $(u^{str})^\alpha(x,t^*_*;\eps)$ is a unique solution of the double ramification hierarchy that satisfies the initial condition $\left.(u^{str})^\alpha\right|_{t_*^*=0}=\delta^{\alpha,1}x$.

\begin{lemma}\label{lemma:property of the string solution}
We have $\left.(u^{str})^\alpha\right|_{t_{\ge 1}^*=0}=t^\alpha_0+\delta^{\alpha,1}x$.
\end{lemma}
\begin{proof}
Consider formal variables $v^\alpha_d$, $1\le\alpha\le N, d\ge 0$, and let $u^\alpha_d=v^\alpha_{d+1}$. Consider the following system of evolutionary PDEs:
\begin{gather}\label{eq:integrated hierarchy}
\frac{\d v^\alpha}{\d t^\beta_q}=\eta^{\alpha\mu}\frac{\delta\og_{\beta,q}}{\delta u^\mu}.
\end{gather}
From the compatibility of the flows of the double ramification hierarchy it easily follows that the system~\eqref{eq:integrated hierarchy} is also compatible. It means that it has a unique solution for an arbitrary polynomial initial condition $\left.v^\alpha\right|_{t^*_*=0}=P^\alpha(x)$. 
Let $(v^{str})^\alpha(x,t^*_*;\eps)$ be a unique solution that satisfies the initial condition $\left.(v^{str})^\alpha\right|_{t^*_*=0}=\delta^{\alpha,1}\frac{x^2}{2}$. We claim that we have the following equation:
\begin{gather}\label{eq:string for integrated}
\frac{\d(v^{str})^\alpha}{\d t^1_0}-\sum_{n\ge 0}t^\gamma_{n+1}\frac{\d(v^{str})^\alpha}{\d t^\gamma_n}=t^\alpha_0+\delta^{\alpha,1}x.
\end{gather}
It can be proved in a way very similar to the proof of Lemma~4.7 in~\cite{Bur14}. We obviously have $(u^{str})^\alpha=\frac{\d(v^{str})^\alpha}{\d t^1_0}$. If we set $t^*_{\ge 1}=0$ in equation~\eqref{eq:string for integrated}, we get the statement of the lemma.
\end{proof}


\subsection{Divisor equation for the Hamiltonians}\label{subsection:divisor equation for the Hamiltonians}

In this section we consider the cohomological field theory associated to the Gromov-Witten theory of some target variety $V$ with vanishing odd cohomology, $H^{odd}(V;\mbC)=0$. Consider the associated double ramification hierarchy. We will use the same notations as in Section~\ref{subsection:divisor recursion}. Recall that we denoted by $e_{\gamma_1},\ldots,e_{\gamma_r}$ a basis in $H^2(V;\mbC)$.

\begin{lemma}\label{lemma:divisor for dr}
For any $i=1,\ldots,r$ and $d\ge 0$, we have $\frac{\d\og_{\alpha,d}}{\d u^{\gamma_i}}=c^\mu_{\alpha\gamma_i}\og_{\mu,d-1}+\left<e_{\gamma_i},q\frac{\d}{\d q}\og_{\alpha,d}\right>$.
\end{lemma}
\begin{proof}
By Theorem~\ref{divisor}, we have $\left<e_{\gamma_i},q\frac{\d}{\d q}g_{\alpha,d}\right>=\d_x^{-1}\frac{\d g_{\alpha,d-1}}{\d t^{\gamma_i}_0}-c^\mu_{\alpha\gamma_i}g_{\mu,d-1}$. Using Theorem~\ref{trr} we obtain
$$
\left<e_{\gamma_i},q\frac{\d}{\d q}g_{\alpha,d}\right>=\frac{\d g_{\alpha,d}}{\d u^{\gamma_i}_0}-c^\mu_{\alpha\gamma_i}g_{\mu,d-1}.
$$
Projecting the both sides of this equation to the space of local functionals we get the statement of the lemma.
\end{proof}


\subsection{Divisor equation for the string solution}\label{subsection:divisor equation for the string solution}

Here we work under the same assumptions, as in the previous section. Consider the string solution $(u^{str})^\alpha$ of the double ramification hierarchy. 

\begin{lemma}\label{lemma:divisor for the string solution}
For any $i=1,2,\ldots,r$, we have 
$$
\frac{\d(u^{str})^\alpha}{\d t^{\gamma_i}_0}-\left<e_{\gamma_i},q\frac{\d(u^{str})^\alpha}{\d q}\right>-\sum_{d\ge 0} c^\mu_{\nu\gamma_i}t^\nu_{d+1}\frac{\d(u^{str})^\alpha}{\d t^\mu_d}=\delta^{\alpha,\gamma_i}.
$$
\end{lemma}
\begin{proof}
Introduce an operator $O_{\gamma_i}$ by $O_{\gamma_i}:=\frac{\d}{\d t^{\gamma_i}_0}-\left<e_{\gamma_i},q\frac{\d}{\d q}\right>-\sum_{d\ge 0} c^\mu_{\nu\gamma_i}t^\nu_{d+1}\frac{\d}{\d t^\mu_d}$. From Lemma~\ref{lemma:property of the string solution} it follows that 
\begin{gather}\label{eq:initial condition for the divisor equation}
\left.O_{\gamma_i}(u^{str})^\alpha\right|_{t^*_*=0}=\delta^{\alpha,\gamma_i}.
\end{gather}
Let $f^\alpha_{\beta,q}:=\eta^{\alpha\mu}\d_x\frac{\delta\og_{\beta,q}}{\delta u^\mu}$. For any $d\ge 0$, we have
\begin{multline}\label{eq:divisor for string solution}
\frac{\d}{\d t^\beta_d}\left(O_{\gamma_i}(u^{str})^{\alpha}\right)=-c^\nu_{\beta\gamma_i}\frac{\d(u^{str})^\alpha}{\d t^\nu_{d-1}}+O_{\gamma_i}\frac{\d(u^{str})^\alpha}{\d t^\beta_d}=-c^\nu_{\beta\gamma_i}f^\alpha_{\nu,d-1}+O_{\gamma_i}f^\alpha_{\beta,d}=\\
=-c^\nu_{\beta\gamma_i}f^\alpha_{\nu,d-1}-\left<e_{\gamma_i},q\frac{\d}{\d q}f^\alpha_{\beta,d}\right>+\sum_{n\ge 0}\frac{\d f^\alpha_{\beta,d}}{\d u^\gamma_n}\d_x^n O_{\gamma_i}(u^{str})^\gamma,
\end{multline}
where we, by definition, put $\frac{\d}{\d t^\nu_{-1}}:=0$. The resulting system of equations can be considered as a system of evolutionary partial differential equations for the power series~$O_{\gamma_i}(u^{str})^\alpha$. Together with the initial condition~\eqref{eq:initial condition for the divisor equation}, it uniquely determines the power series~$O_{\gamma_i}(u^{str})^\alpha$. Lemma~\ref{lemma:divisor for dr} implies that, if we substitute $O_{\gamma_i}(u^{str})^\alpha=\delta^{\alpha,\gamma_i}$ on the right-hand side of~\eqref{eq:divisor for string solution}, we get zero. The lemma is proved.
\end{proof}


\section{Dubrovin-Zhang hierarchy for $\CP1$}

The main goal of this section is to recall the explicit description of the Dubrovin-Zhang hierarchy for~$\CP1$ obtained in~\cite{DZ04}. In Section~\ref{subsection:brief recall} we say a few words about the general theory of the Dubrovin-Zhang hierarchies. We recall the notion of a Miura transformation and also write an explicit formula that relates the descendant and the ancestor Dubrovin-Zhang hierarchies for $\CP1$. In Section~\ref{subsection:extended Toda hierarchy} we review the construction of the extended Toda hierarchy and its relation to the descendant Dubrovin-Zhang hierarchy for $\CP1$. In Section~\ref{subsection:explicit computations} we list some explicit formulae for the ancestor hierarchy that we will use in Section~\ref{section:dr hierarchy for cp1}.

\subsection{Brief recall of the Dubrovin-Zhang theory}\label{subsection:brief recall}

\subsubsection{General theory}

The main reference for the Dubrovin-Zhang theory is the paper~\cite{DZ05}. The theory was later generalized in~\cite{BPS12b}. In this section we follow the approach from~\cite{BPS12b} (see also~\cite{BPS12a}). 

The Dubrovin-Zhang hierarchies form a certain subclass in the class of hamiltonian hierarchies of PDEs~\eqref{eq:Hamiltonian system}. They are associated to semisimple potentials of Gromov-Witten type. Let us describe the family of these potentials. First of all, there is a family of all cohomological field theories. To any cohomological field theory one can associate the so-called ancestor potential, that is defined as the generating series of the correlators of the cohomological field theory. The semisimplicity condition means that a certain associative commutative algebra, associated to the cohomological field theory, doesn't have nilpotents. Ancestor potentials form a subfamily in the family of all potentials of Gromov-Witten type. Given an ancestor potential, there is so-called Givental's $s$-action (or the action of the lower triangular Givental group, see e.g.~\cite{FSZ10}) that produces a family of potentials that correspond to this ancestor potential. These potentials are sometimes called the descendant potentials corresponding to the given ancestor potential. The resulting family of potentials is called the family of potentials of Gromov-Witten type.

The Dubrovin-Zhang hierarchy corresponding to an ancestor potential will be called the ancestor hierarchy, while the hierarchy corresponding to a descendant potential will be called the descendant hierarchy. It is not hard to write explicitly a relation between them. This was done in~\cite{BPS12b} (see also~\cite{BPS12a}). We will write this relation in the case of $\CP1$, see Lemma~\ref{lemma:relation} below.

\subsubsection{The descendant and the ancestor potentials of $\CP1$}

Let us describe some details and also fix notations in the case of $\CP1$. We use the notations from Section~\ref{subsection:divisor recursion}. 

Let $V:=H^*(\CP1;\mbC)$. The semigroup $E\subset H_2(\CP1;\mbZ)$ is generated by the fundamental class~$[\CP1]$, so it is naturally isomorphic to~$\mbZ_{\ge 0}$. The Novikov ring~$\mcN$ is isomorphic to~$\mbC[[q]]$. Consider the Gromov-Witten theory of~$\CP1$. Let $c_{g,n,d}\colon V^{\otimes n}\to H^*(\oM_{g,n};\mbC)$ be the associated cohomological field theory. Let $1,\omega\in H^*(\CP1;\mbC)$ be the unit and the class dual to a point. The matrix of the metric in this basis will be denoted by $\eta=(\eta_{\alpha\beta})_{\alpha,\beta\in\{1,\omega\}}$. 

The ancestor correlators are defined by 
$$
\<\tau_{d_1}(\alpha_1)\tau_{d_2}(\alpha_2)\ldots\tau_{d_n}(\alpha_n)\>_{g,d}:=\int_{\oM_{g,n}} c_{g,n,d}(\otimes_{i=1}^n\alpha_i)\prod_{i=1}^n\psi_i^{d_i},\quad \alpha_i\in V,\quad d,d_i\ge 0.
$$
Introduce variables $t^1_d,t^\omega_d$, $d\ge 0$. The ancestor potential of $\CP1$ is defined by 
\begin{align*}
&F(t;q;\eps):=\sum_{g\ge 0}\eps^{2g} F_g(t;q),\quad\text{where}\\
&F_g(t;q):=\sum_{\substack{n\ge 0\\2g-2+n>0}}\sum_{d\ge 0}\frac{q^d}{n!}\sum_{\substack{\alpha_1,\ldots,\alpha_n\in\{1,\omega\}\\d_1,\ldots,d_n\ge 0}}\left<\prod_{i=1}^n\tau_{d_i}(\alpha_i)\right>_{g,d}\prod_{i=1}^n t_{d_i}^{\alpha_i}.
\end{align*}

As we said, there is the family of descendant potentials corresponding to the ancestor potential~$F$. All these potentials are related by Givental's $s$-action. Among these descendant potentials there is a particular one that also has a simple geometric description. This potential is defined by
\begin{align*}
&F^{desc}(t;q;\eps):=\sum_{g\ge 0}\eps^{2g} F^{desc}_g(t;q),\text{ where}\\
&F^{desc}_g(t;q):=\sum_{n,d\ge 0}\frac{q^d}{n!}\sum_{\substack{\alpha_1,\ldots,\alpha_n\in\{1,\omega\}\\d_1,\ldots,d_n\ge 0}}\left<\prod_{i=1}^n\tau_{d_i}(\alpha_i)\right>^{desc}_{g,d}\prod_{i=1}^n t_{d_i}^{\alpha_i}.
\end{align*}
Recall that by $\left<\prod_{i=1}^n\tau_{d_i}(\alpha_i)\right>^{desc}_{g,d}$ we denote the Gromov-Witten invariants of~$\CP1$.

Let us list several properties of the descendant potential $F^{desc}$. First of all, we have (see e.g.~\cite{Dub96})
\begin{gather}\label{eq:Frobenius potential}
\left.F^{desc}\right|_{\substack{\eps=0\\t^*_{\ge 1=0}}}=\frac{(t^1_0)^2t^\omega_0}{2}+qe^{t^\omega_0}.
\end{gather}
The following two equations are called the string and the divisor equations (see e.g.~\cite{Hor95}):
\begin{align}
&\left(\frac{\d}{\d t^1_0}-\sum_{n\ge 0}t^\alpha_{n+1}\frac{\d}{\d t^\alpha_n}\right)F^{desc}=t^1_0 t^\omega_0,\label{eq:string for descendant potential}\\
&\left(\frac{\d}{\d t^\omega_0}-q\frac{\d}{\d q}-\sum_{n\ge 0} t^1_{n+1}\frac{\d}{\d t^\omega_n}\right)F^{desc}=\frac{(t^1_0)^2}{2}-\frac{\eps^2}{24}.\label{eq:divisor for descendant potential}
\end{align}

Let us write the relation between the potentials~$F$ and~$F^{desc}$ in terms of Givental's $s$-action. The general formula is given in~\cite{Giv01}. Here we adapt it for the case of $\CP1$. Introduce matrices~$S_k, k\ge 0$, by
\begin{align*}
&S_0=Id,\qquad (S_i)^\beta_\alpha:=\sum_{d\ge 0}\left<\tau_{i-1}(\alpha)\tau_0(\mu)\right>^{desc}_{0,d}\eta^{\mu\beta}q^d,\quad i\ge 1,\\
&S(z):=1+\sum_{n\ge 1}S_n z^{-n}.
\end{align*}
Using formulae~\eqref{eq:Frobenius potential},~\eqref{eq:string for descendant potential},~\eqref{eq:divisor for descendant potential} and the so-called topological recursion relation in genus~$0$ (see e.g.~\cite{Get98}), one can quickly compute that, for $k\ge 1$, we have
\begin{gather}\label{eq:S-matrix}
(S_{2k-1})^\alpha_\beta=
\begin{cases}
\frac{1}{k!(k-1)!}q^k,&\text{if $\alpha=1$, $\beta=\omega$};\\
-\frac{2H_{k-1}}{((k-1)!)^2}q^{k-1},&\text{if $\alpha=\omega$, $\beta=1$};\\
0,&\text{otherwise}.
\end{cases}
\qquad
(S_{2k})^\alpha_\beta=
\begin{cases}
\left(\frac{1}{(k!)^2}-\frac{2H_k}{k!(k-1)!}\right)q^k,&\text{if $\alpha=\beta=1$};\\
\frac{1}{(k!)^2}q^k,&\text{if $\alpha=\beta=\omega$};\\
0,&\text{otherwise}.
\end{cases}
\end{gather}
Here $H_k:=1+\frac{1}{2}+\ldots+\frac{1}{k}$, if $k\ge 1$, and $H_0:=0$. Introduce matrices~$s_k, k\ge 1$, by $s(z)=\sum_{n\ge 1} s_n z^{-n}:=\log S(z)$. Let $(s_k)_{\alpha\beta}:=(s_k)_\alpha^\mu\eta_{\mu\beta}$. Then the potentials $F^{desc}$ and $F$ are related by
\begin{align*}
&\exp\left(F^{desc}\right)=\exp\left(\widehat{s(z)}\right)\exp\left(F\right),\quad\text{where}\\
&\widehat{s(z)}:=-\frac{1}{2}(s_3)_{1,1}+\sum_{d\ge 0}(s_{d+2})_{\alpha,1}t^\alpha_d+\frac{1}{2}\sum_{d_1,d_2\ge 0}(-1)^{d_2}(s_{d_1+d_2+1})_{\mu_1\mu_2}t^{\mu_1}_{d_1}t^{\mu_2}_{d_2}+\sum_{\substack{l\ge 1\\d\ge 0}}(s_l)^\mu_\nu t^\nu_{d+l}\frac{\d}{\d t^\mu_d}.
\end{align*}

\subsubsection{Miura transformations in the theory of hamiltonian hierarchies}

Here we want to discuss changes of variables in the theory of hamiltonian hierarchies. We recommend the reader the paper~\cite{DZ05} for a more detailed introduction to this subject. 

First of all, let us modify our notations a little bit. Recall that by $\mcA$ we denoted the ring of differential polynomials in the variables $u^1,\ldots,u^N$. Since we are going to consider rings of differential polynomials in different variables, we want to see the variables in the notation. So for the rest of the paper we denote by $\mcA_{u^1,\ldots,u^N}$ the ring of differential polynomials in the variables $u^1,\ldots,u^N$. The same notation is adopted for the extension $\hcA_{u^1,\ldots,u^N}$ and for the spaces of local functionals~$\Lambda_{u^1,\ldots,u^N}$ and~$\hLambda_{u^1,\ldots,u^N}$.

Consider changes of variables of the form
\begin{align}
&\tu^\alpha(u;u_x,u_{xx},\ldots;\eps)=u^\alpha+\sum_{k\ge 1}\eps^k f^\alpha_k(u;u_x,\ldots,u_k),\quad \alpha=1,\ldots,N,\label{eq:Miura transformation}\\
&f^\alpha_k\in\mcA_{u^1,\ldots,u^N},\quad\deg f^\alpha_k=k.\label{eq:degree condition}
\end{align}
They are called Miura transformations (see e.g.~\cite{DZ05}). It is not hard to see that they are invertible.

Any differential polynomial $f(u)\in\hcA_{u^1,\ldots,u^N}$ can be rewritten as a differential polynomial in the new variables $\tu^\alpha$. The resulting differential polynomial is denoted by $f(\tu)$. The last equation in line~\eqref{eq:degree condition} garanties that, if $f(u)\in\hcA_{u^1,\ldots,u^N}^{[d]}$, then $f(\tu)\in\hcA_{\tu^1,\ldots,\tu^N}^{[d]}$. In other words, a Miura transformation defines an isomorphism $\hcA_{u^1,\ldots,u^N}^{[d]}\simeq\hcA_{\tu^1,\ldots,\tu^N}^{[d]}$. In the same way any Miura transformation identifies the spaces of local functionals $\hLambda^{[d]}_{u^1,\ldots,u^N}$ and $\hLambda^{[d]}_{\tu^1,\ldots,\tu^N}$. For any local functional $\oh[u]\in\hLambda^{[d]}_{u^1,\ldots,u^N}$ the image of it under the isomorphism $\hLambda^{[d]}_{u^1,\ldots,u^N}\stackrel{\sim}{\to}\hLambda^{[d]}_{\tu^1,\ldots,\tu^N}$ is denoted by $\oh[\tu]\in\hLambda^{[d]}_{\tu^1,\ldots,\tu^N}$. 

Let us describe the action of Miura transformations on hamiltonian hierarchies. Suppose we have a hamiltonian system
\begin{gather}\label{eq:Hamiltonian system2}
\frac{\d u^\alpha}{\d\tau_i}=K^{\alpha\mu}\frac{\delta\oh_i[u]}{\delta u^\mu},\quad\alpha=1,\ldots,N,\quad i\ge 1,
\end{gather}
defined by a hamiltonian operator $K$ and a sequence of pairwise commuting local functionals $\oh_i[u]\in\hLambda^{[0]}_{u^1,\ldots,u^N}$, $\{\oh_i[u],\oh_j[u]\}_K=0$. Consider a Miura transformation~\eqref{eq:Miura transformation}. Then in the new variables~$\tu^i$ system~\eqref{eq:Hamiltonian system2} looks as follows (see e.g.~\cite{DZ05}):
\begin{align*}
&\frac{\d\tu^\alpha}{\d\tau_i}=\widetilde K^{\alpha\mu}\frac{\delta\oh_i[\tu]}{\delta \tu^\mu},\quad\text{where}\\
&\widetilde K^{\alpha\beta}=\sum_{p,q\ge 0}\frac{\d \tu^\alpha(u)}{\d u^\mu_p}\d_x^p\circ K^{\mu\nu}\circ(-\d_x)^q\circ\frac{\d \tu^\beta(u)}{\d u^\nu_q}.
\end{align*}

\subsubsection{The descendant and the ancestor Dubrovin-Zhang hierarchies for $\CP1$}

The variables of the Dubrovin-Zhang hierarchies will be denoted by $w^\alpha$. Denote by $\oh_{\alpha,p}[w]\in\hLambda^{[0]}_{w^1,w^\omega}\otimes\mbC[q]$, $\alpha\in\{1,\omega\}, p\ge 0$, the Hamiltonians of the ancestor Dubrovin-Zhang hierarchy for $\CP1$ and by $K$ the hamiltonian operator. The Hamiltonians and the hamiltonian operator of the descendant hierarchy will be denoted by~$\oh_{\alpha,p}^{desc}[w]$ and~$K^{desc}$ correspondingly. For convenience, let us also introduce local functionals $\oh^{desc}_{\alpha,-1}[w]$ by $\oh^{desc}_{\alpha,-1}[w]:=\int\eta_{\alpha\mu}w^\mu dx$. For the operator $S_i$, denote by $S^*_i$ the adjoint operator.

\begin{lemma}\label{lemma:relation}
We have $\oh_{\alpha,p}[w]=\sum_{i=0}^{p+1}(-1)^i(S^*_i)^\mu_\alpha\oh^{desc}_{\mu,p-i}[w]$ and $K^{desc}=K$.
\end{lemma}  
\begin{proof}
The lemma easily follows from Theorems~9,~16 in~\cite{BPS12b} and also from the fact that $S(z)S^*(-z)=Id$ (see e.g.~\cite{Giv01}). The reader should also keep in mind that $(S_1)^\alpha_1=(s_1)^\alpha_1=0$.
\end{proof}


\subsection{Extended Toda hierarchy}\label{subsection:extended Toda hierarchy}

In this section we recall the construction of the extended Toda hierarchy and the Miura transformation that relates it to the descendant Dubrovin-Zhang hierarchy for~$\CP1$. We follow the paper~\cite{DZ04}.

\subsubsection{Construction}

Consider formal variables $v^1,v^2$, the ring of differential polynomials~$\hcA_{v^1,v^2}$ and the tensor product~$\hcA_{v^1,v^2}\otimes\mbC[q,q^{-1}]$. For a formal series  
$$
a=\sum_{k\in\mathbb Z}a_k(v;\eps;q)e^{k\eps\d_x},\quad a_k\in\hcA_{v^1,v^2}\otimes\mbC[q,q^{-1}],
$$
let $a_+:=\sum_{k\ge 0}a_k e^{k\eps\d_x}$ and $\Res(a):=a_0$. Consider the operator
$$
L=e^{\eps\d_x}+v^1+qe^{v^2}e^{-\eps\d_x}.
$$
The equations of the extended Toda hierarchy look as follows:
\begin{align*}
&\frac{\d L}{\d t^1_p}=\eps^{-1}\frac{2}{p!}[(L^p(\log L-H_p))_+,L],\\
&\frac{\d L}{\d t^\omega_p}=\eps^{-1}\frac{1}{(p+1)!}[(L^{p+1})_+,L].
\end{align*}
We refer the reader to~\cite{DZ04} for the precise definition of the logarithm $\log L$. The hamiltonian structure of the extended Toda hierarchy is given by the operator  
\begin{gather}\label{eq:Toda operator}
K^{Td}=
\begin{pmatrix}
0                          & \eps^{-1}(e^{\eps\d_x}-1) \\
\eps^{-1}(1-e^{-\eps\d_x}) & 0
\end{pmatrix}
\end{gather}
and the Hamiltonians
\begin{align}
&\oh^{Td}_{1,p}[v]=\int\left(\frac{2}{(p+1)!}\Res(L^{p+1}(\log L-H_{p+1}))\right)dx,\label{eq:1 Toda Hamiltonian}\\
&\oh^{Td}_{\omega,p}[v]=\int\left(\frac{1}{(p+2)!}\Res(L^{p+2})\right)dx.\label{eq:omega Toda Hamiltonian}
\end{align}
So the equations of the extended Toda hierarchy can be written as follows:
$$
\frac{\d v^\alpha}{\d t^\beta_p}=(K^{Td})^{\alpha\mu}\frac{\delta\oh^{Td}_{\beta,p}[v]}{\delta v^\mu}.
$$

\subsubsection{Descendant Dubrovin-Zhang hierarchy for $\CP1$}

In~\cite{DZ04} B.~Dubrovin and Y.~Zhang proved the following theorem.
\begin{theorem}\label{theorem:DZ theorem}
The descendant hierarchy for~$\CP1$ is related to the extended Toda hierarchy by the Miura transformation 
\begin{gather}\label{eq:w-v relation}
w^1(v)=\frac{\eps\d_x}{e^{\eps\d_x}-1}v^1,\qquad w^\omega(v)=\frac{\eps^2\d_x^2}{e^{\eps\d_x}+e^{-\eps\d_x}-2}v^2.
\end{gather}
\end{theorem}

\begin{rem}
The construction of the Dubrovin-Zhang hierarchy (see~\cite{BPS12b}) immediately implies that the Hamiltonians $\oh^{desc}_{\alpha,p}[w]$ contain only nonnegative powers of $q$. The fact that $\oh_{\omega,p}^{Td}[v]\in\hLambda^{[0]}_{v^1,v^2}\otimes\mbC[q]$ easily follows from formula~\eqref{eq:omega Toda Hamiltonian}. The fact that $\oh_{1,p}^{Td}[v]\in\hLambda^{[0]}_{v^1,v^2}\otimes\mbC[q]$ is not so trivial, since the coefficients of the logarithm~$\log L$ contain negative powers of~$q$. We will show how to derive it from~\eqref{eq:1 Toda Hamiltonian} in Section~\ref{subsection:proof of DZ Hamiltonian in degree 0}.
\end{rem}


\subsection{Several computations for the ancestor hierarchy}\label{subsection:explicit computations}

The ancestor Dubrovin-Zhang hierarchy for~$\CP1$ comes with a specific solution (see e.g.~\cite{BPS12b}):
$$
(w^{top})^\alpha(x,t;\eps;q):=\eta^{\alpha\mu}\left.\frac{\d^2 F}{\d t^\mu_0\d t^1_0}\right|_{t^\alpha_d\mapsto t^\alpha_d+\delta^{\alpha,1}\delta_{d,0}x}.
$$ 
It is called the topological solution. 

Now for the rest of the paper we fix Miura transformation~\eqref{eq:Miura}. Let
$$
(u^{top})^\alpha(x,t;\eps;q):=\left.u^\alpha(w)\right|_{w^\nu_i=\d_x^i(w^{top})^\nu}=\frac{e^{\frac{\eps}{2}\d_x}-e^{-\frac{\eps}{2}\d_x}}{\eps\d_x}(w^{top})^\alpha.
$$

\subsubsection{The string and the divisor equations for the topological solution}

We have the following equations for the ancestor potential $F$ (see e.g.~\cite{KM94}):
\begin{align}
&\left(\frac{\d}{\d t^1_0}-\sum_{n\ge 0}t^\alpha_{n+1}\frac{\d}{\d t^\alpha_n}\right)F=t^1_0 t^\omega_0,\label{eq:string for ancestor}\\
&\left(\frac{\d}{\d t^\omega_0}-q\frac{\d}{\d q}-\sum_{d\ge 0}t^1_{d+1}\frac{\d}{\d t^\omega_d}-q\sum_{d\ge 0}t^\omega_{d+1}\frac{\d}{\d t^1_d}\right)F=\frac{(t^1_0)^2}{2}-\frac{\eps^2}{24}.\label{eq:divisor for ancestor}
\end{align}
The second equation is an analog of the divisor equation~\eqref{eq:divisor for descendant potential} for the descendant potential~$F^{desc}$. In order to derive it, one should use the following formula:
\begin{gather}\label{eq:three-point}
c_{\omega\alpha\beta}=
\begin{cases}
1,&\text{if $\alpha=\beta=1$};\\
q,&\text{if $\alpha=\beta=\omega$};\\
0,&\text{otherwise}.
\end{cases}
\end{gather}
It can be easily checked using~\eqref{eq:Frobenius potential}.

From equation~\eqref{eq:string for ancestor} it immediately follows that
\begin{gather}\label{eq:initial condition for DZ}
\left.(u^{top})^\alpha\right|_{t^*_*=0}=\delta^{\alpha,1}x.
\end{gather}
Equations~\eqref{eq:string for ancestor} and~\eqref{eq:divisor for ancestor} also imply that
\begin{align}
&\left(\frac{\d}{\d t^1_0}-\sum_{n\ge 0}t^\alpha_{n+1}\frac{\d}{\d t^\alpha_n}\right)(u^{top})^\alpha=\delta^{\alpha,1},\label{eq:string for topological}\\
&\left(\frac{\d}{\d t^\omega_0}-q\frac{\d}{\d q}-\sum_{d\ge 0}t^1_{d+1}\frac{\d}{\d t^\omega_d}-q\sum_{d\ge 0}t^\omega_{d+1}\frac{\d}{\d t^1_d}\right)(u^{top})^\alpha=\delta^{\alpha,\omega}.\label{eq:divisor for topological}
\end{align}

\subsubsection{Some Hamiltonians of the ancestor hierarchy}

Let $S(z):=\frac{e^{\frac{z}{2}}-e^{-\frac{z}{2}}}{z}$. The following properties will be crucial for the proof of Theorem~\ref{theorem:dr for cp1}:
\begin{align}
\left.\oh_{\omega,p}[u]\right|_{q=0}=&\int\frac{(u^1)^{p+2}}{(p+2)!}dx,\label{eq:DZ omega Hamiltonian in degree 0}\\
\left.\oh_{1,p}[u]\right|_{q=0}=&\int\left(\frac{(u^1)^{p+1}u^\omega}{(p+1)!}+\sum_{g\ge 1}\eps^{2g}r_{p,g}(u^1)\right)dx,\quad r_{p,g}\in\mcA_{u^1},\quad\deg r_{p,g}=2g,\label{eq:DZ 1 Hamiltonian in degree 0}\\
\left.\oh_{1,1}[u]\right|_{q=0}=&\int\left(\frac{(u^1)^2 u^\omega}{2}+\sum_{g\ge 1}\eps^{2g}\frac{B_{2g}}{(2g)!}u^1u^1_{2g}\right)dx,\label{eq:DZ 1,1 Hamiltonian in degree 0}\\
\oh_{\omega,0}[u]=&\int\left(\frac{(u^1)^2}{2}+q\left(e^{S(\eps\d_x)u^\omega}-u^\omega\right)\right)dx.\label{eq:DZ omega Hamiltonian}
\end{align}
We will prove these formulae in Appendix~\ref{section:computations}.


\section{Double ramification hierarchy for $\CP1$}\label{section:dr hierarchy for cp1}

In this section we prove Theorem~\ref{theorem:dr for cp1}. First of all, let us consider the hamiltonian structures. By Lemma~\ref{lemma:relation}, we have $K=K^{desc}$. The fact, that Miura transformation~\eqref{eq:Miura} transforms the operator~$K^{desc}$ to $\eta\d_x$, was observed in~\cite{DZ05}. This is actually an easy computation. By~\eqref{eq:Toda operator} and Theorem~\ref{theorem:DZ theorem}, Miura transformation~\eqref{eq:Miura} transforms the operator $K^{desc}$ to the operator~$\widetilde K$, where
$$
\widetilde K^{\alpha\beta}=\sum_{p,q\ge 0}\frac{\d u^\alpha(v)}{\d v^\mu_p}\d_x^p\circ (K^{Td})^{\mu\nu}\circ(-\d_x)^q\circ\frac{\d u^\beta(v)}{\d v^\nu_q}=\eta^{\alpha\beta}\d_x.
$$
We conclude that Miura transformation~\eqref{eq:Miura} transforms the hamiltonian operator of the ancestor hierarchy for~$\CP1$ to the hamiltonian operator of the double ramification hierarchy for~$\CP1$.

It remains to prove that $\og_{\alpha,p}[u]=\oh_{\alpha,p}[u]$. The proof is splitted in three steps. First, in Section~\ref{subsection:degree 0 parts} we prove this equation in degree zero: $\og_{\alpha,p}[u]|_{q=0}=\oh_{\alpha,p}[u]|_{q=0}$. Then in Section~\ref{subsection:omega flow} we prove that $\og_{\omega,0}[u]=\oh_{\omega,0}[u]$. Finally, in Section~\ref{subsection:final step} we show that this information, together with the string and the divisor equations, is enough to prove that $u^{top}(x,t;\eps;q)=u^{str}(x,t;\eps;q)$. After that it is very easy to show that~$\og_{\alpha,p}[u]=\oh_{\alpha,p}[u]$.


\subsection{Degree $0$ parts}\label{subsection:degree 0 parts}

In this section we prove that
\begin{gather}\label{eq:degree zero parts}
\left.\og_{\alpha,d}[u]\right|_{q=0}=\left.\oh_{\alpha,d}[u]\right|_{q=0}.
\end{gather}
The degree zero part of our cohomological field theory can be described very explicitly. For any $g,a,b\ge 0$, $2g-2+a+b>0$, we have (see e.g.~\cite{GP98})
\begin{gather}\label{eq:cohft in degree 0}
c_{g,a+b,0}(1^{\otimes a}\otimes\omega^{\otimes b})=
\begin{cases}
2(-1)^{g-1}\lambda_{g-1},&\text{if $b=0$},\\
(-1)^g\lambda_g,&\text{if $b=1$},\\
0,&\text{otherwise}.
\end{cases}
\end{gather}
Recall the following fact (see e.g.~\cite{Bur14}).
\begin{lemma}\label{lemma:genus zero parts}
We have $\left.\og_{\alpha,d}[u]\right|_{\eps=0}=\left.\oh_{\alpha,d}[u]\right|_{\eps=0}$.
\end{lemma}
\noindent In Section~\ref{subsubsection:omega parts} we prove equation~\eqref{eq:degree zero parts} for $\alpha=\omega$ and in Section~\ref{subsubsection:unit parts} we prove it for $\alpha=1$.

\subsubsection{Hamiltonians $\left.\og_{\omega,d}[u]\right|_{q=0}$ and $\left.\oh_{\omega,d}[u]\right|_{q=0}$}\label{subsubsection:omega parts}
Here we prove that $\left.\og_{\omega,d}[u]\right|_{q=0}=\left.\oh_{\omega,d}[u]\right|_{q=0}$.
\begin{lemma}\label{lemma:DR omega Hamiltonian in degree zero}
We have $\left.\og_{\omega,d}[u]\right|_{q=0}=\int\frac{(u^1)^{d+2}}{(d+2)!}dx$.
\end{lemma}
\begin{proof}
For any $g\ge 1$, we have $\lambda_g^2=0$. Therefore, from~\eqref{eq:cohft in degree 0} we can immediately conclude that
\begin{gather}\label{eq:omega hamiltonian in degree 0}
\left.\og_{\omega,d}[u]\right|_{q=0}=\left.\og_{\omega,d}[u]\right|_{q=\eps=0}\stackrel{\substack{\text{by Lemma~\ref{lemma:genus zero parts}}\\\text{and eq.~\eqref{eq:DZ omega Hamiltonian in degree 0}}}}{=}\int\frac{(u^1)^{d+2}}{(d+2)!}dx.
\end{gather}
\end{proof}
By~\eqref{eq:DZ omega Hamiltonian in degree 0}, we get $\left.\og_{\omega,d}[u]\right|_{q=0}=\left.\oh_{\omega,d}[u]\right|_{q=0}$.

\subsubsection{Hamiltonians $\left.\og_{1,d}[u]\right|_{q=0}$ and $\left.\oh_{1,d}[u]\right|_{q=0}$}\label{subsubsection:unit parts}

The goal of this section is to prove that
\begin{gather}\label{eq:DZ1 and DR1 Hamiltonians in degree 0}
\left.\og_{1,d}[u]\right|_{q=0}=\left.\oh_{1,d}[u]\right|_{q=0}.
\end{gather}
The proof of this fact is not so direct, as in the previous section. We begin with the following lemma.

\begin{lemma}\label{lemma:DR 1 Hamiltonians in degree 0}
The local functionals $\left.\og_{1,d}[u]\right|_{q=0}$ have the following form
\begin{gather*}
\left.\og_{1,d}[u]\right|_{q=0}=\int\left(\frac{(u^1)^{d+1}u^\omega}{(d+1)!}+\sum_{g\ge 1}\eps^{2g}f_{d,g}(u^1)\right)dx,
\end{gather*}
for some differential polynomials $f_{d,g}\in\mcA_{u^1}$, $\deg f_{d,g}=2g$.
\end{lemma}
\begin{proof}
By Lemma~\ref{lemma:genus zero parts} and equation~\eqref{eq:DZ 1 Hamiltonian in degree 0}, we have $\left.\og_{1,d}[u]\right|_{q=\eps=0}=\int\frac{(u^1)^{d+1}u^\omega}{(d+1)!}dx$. 
Using~\eqref{eq:cohft in degree 0} and the fact that, for $g\ge 1$, $\lambda_g^2=0$, we obtain
\begin{gather}\label{eq:g1d general form}
\left.\og_{1,d}[u]\right|_{q=0}=\int\frac{(u^1)^{d+1}u^\omega}{(d+1)!}dx-2\sum_{\substack{g\ge 1\\n\ge 2}}\frac{\eps^{2g}}{n!}\sum_{\substack{a_1,\ldots,a_n\in\mbZ\\\sum a_i=0}}\left(\int_{\DR_g(0,a_1,\ldots,a_n)}\psi_1^d\lambda_g\lambda_{g-1}\right)\prod_{i=1}^n p^1_{a_i}.
\end{gather}
Note that the sum on the right-hand side of this equation contains only monomials with the variables~$p^1_i$. The lemma is now clear.
\end{proof}

\begin{lemma}\label{lemma:DR 1,1 Hamiltonian in degree 0}
We have $\left.\og_{1,1}[u]\right|_{q=0}=\int\left(\frac{(u^1)^2 u^\omega}{2}+\sum_{g\ge 1}\eps^{2g}\frac{B_{2g}}{(2g)!}u^1u^1_{2g}\right)dx$.
\end{lemma}
\begin{proof}
By~\eqref{eq:g1d general form}, we have
$$
\left.\og_{1,1}[u]\right|_{q=0}=\int\frac{(u^1)^2 u^\omega}{2}dx-\sum_{g\ge 1}\eps^{2g}\sum_{a\in\mbZ}\left(\int_{\DR_g(0,a,-a)}\psi_1\lambda_g\lambda_{g-1}\right) p^1_ap^1_{-a}.
$$
For $g\ge 1$, we have $\int_{\DR_g(0,a,-a)}\psi_1\lambda_g\lambda_{g-1}=2g\int_{\DR_g(a,-a)}\lambda_g\lambda_{g-1}=a^{2g}\frac{|B_{2g}|}{(2g)!}$, where the computation of the last integral can be found, for example, in~\cite{CMW12}. We obtain
$$
\left.\og_{1,1}[u]\right|_{q=0}=\int\left(\frac{(u^1)^2 u^\omega}{2}+\sum_{g\ge 1}\eps^{2g}\frac{B_{2g}}{(2g)!}u^1u^1_{2g}\right)dx.
$$
The lemma is proved.
\end{proof}

After this preparation we are ready for proving equation~\eqref{eq:DZ1 and DR1 Hamiltonians in degree 0}. Let 
$$
f_d(u^1;\eps):=\sum_{g\ge 1}\eps^{2g}f_{d,g}(u^1)\in\hcA^{[0]}_{u^1}.
$$
Let us expand the relation $\{\og_{1,d}[u],\og_{1,1}[u]\}_{\eta\d_x}=0$ using Lemmas~\ref{lemma:DR 1 Hamiltonians in degree 0} and~\ref{lemma:DR 1,1 Hamiltonian in degree 0}. We get
\begin{gather}\label{eq:equation for 1-Hamiltonians}
\int\left(\frac{\delta f_d}{\delta u^1}\d_x\left(\frac{(u^1)^2}{2}\right)+\frac{(u^1)^{d+1}}{(d+1)!}\d_x\left(2\sum_{g\ge 1}\eps^{2g}\frac{B_{2g}}{(2g)!}u^1_{2g}\right)\right)dx.
\end{gather}
Introduce a local functional $\os_d[u^1]\in\hLambda^{[0]}_{u^1}$ by $\os_d[u^1]:=\int f_d(u^1;\eps)dx$. Equation~\eqref{eq:equation for 1-Hamiltonians} can be rewritten as follows:
\begin{gather*}
\left\{\os_d,\int\frac{(u^1)^3}{6}dx\right\}_{\d_x}+2\int\left(\frac{(u^1)^{d+1}}{(d+1)!}\sum_{g\ge 1}\eps^{2g}\frac{B_{2g}}{(2g)!}u^1_{2g+1}\right)dx=0.
\end{gather*}
From~\cite[Lemma~2.5]{Bur13} it follows that this equation uniquely determines the local functional~$\os_d[u^1]$ and, therefore, the Hamiltonian~$\og_{1,d}[u]|_{q=0}$. Observe that the same argument works for the Dubrovin-Zhang Hamiltonians $\oh_{1,d}[u]$. Equation~\eqref{eq:DZ 1 Hamiltonian in degree 0} says that the Hamiltonian~$\oh_{1,d}[u]|_{q=0}$ has the same form as the Hamiltonian~$\og_{1,d}[u]|_{q=0}$. Moreover, by~\eqref{eq:DZ 1,1 Hamiltonian in degree 0}, we have $\og_{1,1}[u]|_{q=0}=\oh_{1,1}[u]|_{q=0}$. Since Miura transformation~\eqref{eq:Miura} transforms the hamiltonian operator of the ancestor Dubrovin-Zhang hierarchy to the operator~$\eta\d_x$, we have $\{\oh_{1,d}[u],\oh_{1,1}[u]\}_{\eta\d_x}=0$. We conclude that $\og_{1,d}[u]|_{q=0}=\oh_{1,d}[u]|_{q=0}$.


\subsection{Hamiltonian $\og_{\omega,0}[u]$}\label{subsection:omega flow}

The goal of this section is to prove that
\begin{gather}\label{eq:Hamiltonian g0,omega}
\og_{\omega,0}[u]=\int\left(\frac{(u^1)^2}{2}+q\left(e^{S(\eps\d_x)u^\omega}-u^\omega\right)\right)dx.
\end{gather}
We start with the following lemma.
\begin{lemma}\label{lemma:omega,0-Hamiltonian:general form}
The Hamiltonian $\og_{\omega,0}$ has the form $\og_{\omega,0}[u]=\int\left(\frac{(u^1)^2}{2}+q f(u^\omega;\eps)\right)dx$, for a differential polynomial $f\in\hcA^{[0]}_{u^\omega}$ such that $\left.\frac{\d f}{\d u^\omega_i}\right|_{u^\omega_*=0}=0$.
\end{lemma}
\begin{proof}
Denote by~$\deg$ the cohomological degree. We have the following formula (\cite{KM94}):
\begin{gather}\label{eq:dimension}
\deg c_{g,n,d}(\gamma_1\otimes\ldots\otimes\gamma_n)=2(g-1-2d)+\sum_i\deg\gamma_i,\quad \gamma_i\in\{1,\omega\}.
\end{gather}
In order to compute the Hamiltonian $\og_{0,\omega}$, we have to compute the integrals
$$
\int_{\DR_g(0,a_1,\ldots,a_n)}\lambda_g c_{g,n+1,d}(\omega\otimes \gamma_1\otimes\ldots\otimes\gamma_n).
$$
From~\eqref{eq:dimension} it follows that this integral can be nonzero, only if 
\begin{gather}\label{eq:counting dimensions}
\half\sum_{i=1}^n\deg\gamma_i=n-2+2d.
\end{gather}
Since $\deg\gamma_i\le 2$, we get $d\le 1$. The case $d=0$ is described by Lemma~\ref{lemma:DR omega Hamiltonian in degree zero}. Suppose $d=1$. Then equation~\eqref{eq:counting dimensions} immediately implies that $\gamma_1=\gamma_2=\ldots=\gamma_n=\omega$. We get
$$
\og_{\omega,0}[u]=\int\frac{(u^1)^2}{2}dx+q\sum_{\substack{g\ge 0\\n\ge 2}}\frac{(-\eps^2)^g}{n!}\sum_{\substack{a_1,\ldots,a_n\in\mbZ\\\sum a_i=0}}\left(\int_{\DR_g(0,a_1,\ldots,a_n)}\lambda_g c_{g,n+1,1}(\omega^{\otimes n})\right)\prod_{i=1}^np^\omega_{a_i}.
$$
Note that the sum on the right-hand side of this equation contains only monomials with the variables~$p^\omega_i$. The lemma is now clear.
\end{proof}

Let us prove~\eqref{eq:Hamiltonian g0,omega}. By Lemma~\ref{lemma:omega,0-Hamiltonian:general form}, we have
\begin{gather}\label{eq:tmp1}
\og_{\omega,0}[u]=\int\left(\frac{(u^1)^2}{2}+q f(u^\omega;\eps)\right)dx.
\end{gather}
From Lemma~\ref{lemma:simple divisor} and formula~\eqref{eq:three-point} it follows that $\og[u]-\og[u]|_{q=0}=\int q\left(f-\frac{(u^\omega)^2}{2}\right)dx$. By the dilaton equation~\eqref{eq:dilaton for og}, we have  
\begin{multline}\label{eq:tmp2}
\og_{1,1}[u]=(D-2)\og[u]=\og_{1,1}[u]|_{q=0}+\int q(D-2)fdx\stackrel{\text{by Lemma~\ref{lemma:DR 1,1 Hamiltonian in degree 0}}}{=}\\
=\int\left(\frac{(u^1)^2 u^\omega}{2}+\sum_{g\ge 1}\eps^{2g}\frac{B_{2g}}{(2g)!}u^1u^1_{2g}+q(D-2)f\right)dx.\end{multline}
Recall that the operator $D$ is defined by $D:=\sum_{n\ge 0}(n+1)u^\alpha_n\frac{\d}{\d u^\alpha_n}$. We have the equation~$\{\og_{\omega,0},\og_{1,1}\}=0$. Let us write the coefficient of $q$ in this equation using formulae~\eqref{eq:tmp1} and~\eqref{eq:tmp2}. We get
\begin{gather*}
\int\left(u^1\d_x\frac{\delta}{\delta u^\omega}(D-2)f-\left(\d_x\frac{\delta f}{\delta u^\omega}\right)\left(u^1 u^\omega+2\sum_{g\ge 1}\frac{B_{2g}}{(2g)!}u^1_{2g}\right)\right)dx=0.
\end{gather*}
Let us apply the variational derivative~$\frac{\delta}{\delta u^1}$ to the left-hand side of this expression. We obtain
\begin{gather}\label{eq:relation for f}
\d_x\frac{\delta}{\delta u^\omega}(D-2)f=u^\omega\d_x\frac{\delta f}{\delta u^\omega}+2\sum_{g\ge 1}\frac{B_{2g}}{(2g)!}\d_x^{2g+1}\frac{\delta f}{\delta u^\omega}=0.
\end{gather}
We need to prove that $\int f(u^\omega)dx=\int\left(e^{S(\eps\d_x)u^\omega}-u^\omega\right)dx$. It is sufficient to prove that
\begin{gather}\label{eq:variational derivative for omega}
\frac{\delta f}{\delta u^\omega}\stackrel{?}{=}\frac{\delta}{\delta u^\omega}\left(e^{S(\eps\d_x)u^\omega}-u^\omega\right)=S(\eps\d_x)e^{S(\eps\d_x)u^\omega}-1.
\end{gather}
Let $r(u^\omega;\eps):=\frac{\delta f}{\delta u^\omega}$. Since $[\frac{\delta}{\delta u^\alpha},D]=\frac{\delta}{\delta u^\alpha}$ and $[\d_x,D]=-\d_x$, equation~\eqref{eq:relation for f} implies that
\begin{gather}\label{eq:omega-recursion}
(D-2)\d_x r=u^\omega\d_x r+2\sum_{g\ge 1}\eps^{2g}\frac{B_{2g}}{(2g)!}\d_x^{2g+1}r.
\end{gather}
Let us introduce another grading $\deg_{ord}$ in $\mcA_{u^\omega}$ putting $\deg_{ord}u^\omega_i=1$. Let us decompose the differential polynomial $r$ into the sum of components that are homogeneous with respect to both the differential degree~$\deg$ and the new degree~$\deg_{ord}$:
$$
r_{d,g}=\sum_{\substack{g\ge 0\\d\ge 1}}r_{d,g}\eps^{2g},\quad r_{d,g}\in\mcA_{u^\omega},\quad\deg_{ord} r_{d,g}=d,\quad\deg r_{d,g}=2g.
$$
Let us, by definition, put $r_{0,g}:=0$. Then equation~\eqref{eq:omega-recursion} implies that, for any $d\ge 1$ and $g\ge 0$, we have
\begin{gather}\label{eq:main omega recursion}
(d+2g-1)\d_x r_{d,g}=u^\omega\d_x r_{d-1,g}+2\sum_{g_1=1}^g \frac{B_{2g}}{(2g)!}\d_x^{2g_1+1}r_{d,g-g_1}.
\end{gather}
Recall that the operator $\d_x\colon\mcA_{u^\omega}\to \mcA_{u^\omega}$ vanishes only on constants. Then it is easy to see that equation~\eqref{eq:main omega recursion} allows to reconstruct all differential polynomials $r_{d,g}$ starting from $r_{1,0}$. From~\eqref{eq:three-point} it follows that $r_{1,0}=u^\omega$.

We see that, in order to prove~\eqref{eq:variational derivative for omega}, it remains to check that $r=S(\eps\d_x)e^{S(\eps\d_x)u^\omega}-1$ satisfies equation~\eqref{eq:omega-recursion}. So we have to prove the following identity:
\begin{gather*}
(D-2)\d_x S(\eps\d_x)e^{S(\eps\d_x)u^\omega}\stackrel{?}{=}u\d_x S(\eps\d_x)e^{S(\eps\d_x)u^\omega}+2\sum_{g\ge 1}\eps^{2g}\frac{B_{2g}}{(2g)!}\d_x^{2g+1}S(\eps\d_x)e^{S(\eps\d_x)u^\omega}.
\end{gather*}
In order to shorten the computations a little bit, let us make the rescaling $x\mapsto\eps x$ and denote $u^\omega$ by~$u$. Thus, we have to prove the identity
\begin{gather}\label{eq:identity1}
(D-2)\d_x S(\d_x)e^{S(\d_x)u}\stackrel{?}{=}u\d_x S(\d_x)e^{S(\d_x)u}+2\sum_{g\ge 1}\frac{B_{2g}}{(2g)!}\d_x^{2g+1}S(\d_x)e^{S(\d_x)u}.
\end{gather}
We have $\sum_{g\ge 0}\frac{B_{2g}}{(2g)!}z^{2g}=\frac{z}{2}\frac{e^{\halfz}+e^{-\halfz}}{e^{\halfz}-e^{-\halfz}}$. Therefore,~\eqref{eq:identity1} is equivalent to 
\begin{gather}\label{eq:identity2}
D\d_x S(\d_x)e^{S(\d_x)u}\stackrel{?}{=}u\d_x S(\d_x)e^{S(\d_x)u}+\d_x\frac{e^{\halfdx}+e^{-\halfdx}}{e^{\halfdx}-e^{-\halfdx}}\d_x S(\d_x)e^{S(\d_x)u}.
\end{gather}
It is easy to see that the right-hand side is equal to 
$$
u\left(e^{\frac{\d_x}{2}}-e^{-\halfdx}\right)e^{S(\d_x)u}+\d_x\left(e^{\halfdx}+e^{-\halfdx}\right)e^{S(\d_x)u}.
$$
Let us transform the left-hand side of~\eqref{eq:identity2}. Since $[D,\d_x]=\d_x$, we get 
\begin{align*}
[D,\d_x S(\d_x)]=&\left[D,e^{\halfdx}-e^{-\halfdx}\right]=\frac{\d_x}{2}\left(e^{\halfdx}+e^{-\halfdx}\right),\\
D\d_x S(\d_x)e^{S(\d_x)u}=&\left(e^{\halfdx}-e^{-\halfdx}\right)D e^{S(\d_x)u}+\frac{\d_x}{2}\left(e^{\halfdx}+e^{-\halfdx}\right)e^{S(\d_x)u}.
\end{align*}
We conclude that~\eqref{eq:identity2} is equivalent to
\begin{gather}\label{eq:identity3}
\left(e^{\halfdx}-e^{-\halfdx}\right)D e^{S(\d_x)u}\stackrel{?}{=}u\left(e^{\halfdx}-e^{-\halfdx}\right)e^{S(\d_x)u}+\frac{\d_x}{2}\left(e^{\halfdx}+e^{-\halfdx}\right)e^{S(\d_x)u}.
\end{gather}
We have $D e^{S(\d_x)u}=\frac{1}{2}\left(e^{\halfdx}+e^{-\halfdx}\right)u\cdot e^{S(\d_x)u}$. Therefore, the left-hand side of~\eqref{eq:identity3} is equal to
\begin{gather}\label{last1}
\left(e^{\halfdx}-e^{-\halfdx}\right)\left[\frac{1}{2}\left(e^{\halfdx}+e^{-\halfdx}\right)u\cdot e^{S(\d_x)u}\right]=\half\left(e^{\d_x}+1\right)u\cdot e^{\frac{e^{\d_x}-1}{\d_x}u}-\half\left(1+e^{-\d_x}\right)u\cdot e^{\frac{1-e^{-\d_x}}{\d_x}u}.
\end{gather}
On the other hand, the right-hand side of~\eqref{eq:identity3} is equal to
\begin{align}
&u\left(e^{\halfdx}-e^{-\halfdx}\right)e^{S(\d_x)u}+\half\left(e^{\halfdx}+e^{-\halfdx}\right)\left[\left(e^{\halfdx}-e^{-\halfdx}\right)u\cdot e^{S(\d_x)u}\right]=\notag\\
=&u\left(e^{\frac{e^{\d_x}-1}{\d_x}u}-e^{\frac{1-e^{-\d_x}}{\d_x}u}\right)+\half\left(e^{\d_x}-1\right)u\cdot e^{\frac{e^{\d_x}-1}{\d_x}u}+\half\left(1-e^{-\d_x}\right)u\cdot e^{\frac{1-e^{-\d_x}}{\d_x}u}.\label{last2}
\end{align}
It is easy to see that~\eqref{last1} is equal to~\eqref{last2}. Equation~\eqref{eq:Hamiltonian g0,omega} is finally proved.


\subsection{Final step}\label{subsection:final step}

In this section we prove that $\og_{\alpha,d}[u]=\oh_{\alpha,d}[u]$. Recall that by~$(u^{str})^\alpha(x,t;\eps;q)$ (see Section~\ref{subsection:property of the string solution}) we denote the string solution of the double ramification hierarchy for $\CP1$ and by~$(u^{top})^\alpha(x,t;\eps;q)$ we denote the Miura transform of the topological solution of the ancestor Dubrovin-Zhang hierarchy for $\CP1$ (see Section~\ref{subsection:explicit computations}).

\begin{lemma}
We have $(u^{top})^\alpha(x,t,\eps;q)=(u^{str})^\alpha(x,t;\eps;q)$.
\end{lemma}
\begin{proof}
By the definition of the string solution and equation~\eqref{eq:initial condition for DZ}, the initial conditions agree: $\left.(u^{top})^\alpha\right|_{t^*_*=0}=\left.(u^{str})^\alpha\right|_{t^*_*=0}=\delta^{\alpha,1}x$. Since $\og_{\alpha,d}[u]|_{q=0}=\oh_{\alpha,d}[u]|_{q=0}$, we get
\begin{gather}\label{eq:degree zero parts of the solution}
\left.(u^{top})^\alpha\right|_{q=0}=\left.(u^{str})^\alpha\right|_{q=0}.
\end{gather}
From~\eqref{eq:DZ omega Hamiltonian} it follows that the power series $(u^{top})^\alpha$ is a solution of the following system:
\begin{gather}\label{eq:Toda system}
\left\{
\begin{aligned}
&\frac{\d u^1}{\d t^\omega_0}=q\left(S(\eps\d_x)e^{S(\eps\d_x)u^\omega}-1\right),\\
&\frac{\d u^\omega}{\d t^\omega_0}=u^1_x.
\end{aligned}\right.
\end{gather}
The argument, very similar to the one from the paper~\cite{Pan00} (see the proof of Proposition 2 there), shows that the string equation~\eqref{eq:string for topological}, the divisor equation~\eqref{eq:divisor for topological} and the system~\eqref{eq:Toda system} uniquely determine the power series $(u^{top})^{\alpha}$ starting from the degree $0$ part $(u^{top})^\alpha|_{q=0}$. On the other hand, we can apply the same arguments to the string solution~$(u^{str})^\alpha$ of the double ramification hierarchy. The string equation for~$(u^{str})^\alpha$ was derived in~\cite{Bur14} and it coincides with~\eqref{eq:string for topological}. By Lemma~\ref{lemma:divisor for the string solution} and formulae~\eqref{eq:three-point}, we have the divisor equation for $(u^{str})^\alpha$:
$$
\left(\frac{\d}{\d t^\omega_0}-q\frac{\d}{\d q}-\sum_{d\ge 0}t^1_{d+1}\frac{\d}{\d t^\omega_d}-q\sum_{d\ge 0}t^\omega_{d+1}\frac{\d}{\d t^1_d}\right)(u^{str})^\alpha=\delta^{\alpha,\omega}.
$$
It coincides with~$\eqref{eq:divisor for topological}$. Since $\og_{\omega,0}[u]=\oh_{\omega,0}[u]$, the power series~$(u^{str})^\alpha$ is a solution of system~\eqref{eq:Toda system}. In the same way as~$(u^{top})^\alpha$, we can now uniquely reconstruct $(u^{str})^{\alpha}$ from the degree~$0$ part $(u^{str})^\alpha|_{q=0}$. From~\eqref{eq:degree zero parts of the solution} it follows that $(u^{str})^\alpha=(u^{top})^\alpha$. The lemma is proved.
\end{proof}

Consider the ancestor Dubrovin-Zhang hierarchy for $\CP1$ in the variables $u^\alpha$:
\begin{gather}\label{eq:DZ in u-coordinates}
\frac{\d u^\alpha}{\d t^\beta_d}=\eta^{\alpha\mu}\d_x\frac{\delta\oh_{\beta,d}[u]}{\delta u^\mu}.
\end{gather}
Using the string equation~\eqref{eq:string for ancestor} it is easy to check that
\begin{gather}\label{eq:change of variables}
\left.((u^{top})^\alpha_d\right|_{x=0}=t^\alpha_d+\delta^{\alpha,1}\delta_{d,1}+O(t^2)+O(\eps^2).
\end{gather}
From this equation it follows that any power series in the variables $t^\nu_i,\eps$ and $q$ can be written as a power series in $\left(\left.(u^{top})^\alpha_d\right|_{x=0}-\delta^{\alpha,1}\delta_{d,1}\right)$, $\eps$ and $q$ in a unique way. Since~$(u^{top})^\alpha$ is a solution of~\eqref{eq:DZ in u-coordinates}, we conclude that the differential polynomials $\eta^{\alpha\mu}\d_x\frac{\delta\oh_{\beta,d}[u]}{\delta u^\mu}$ can be uniquely reconstructed from the power series $(u^{top})^\alpha$. Therefore, the variational derivatives $\frac{\delta\oh_{\beta,d}[u]}{\delta u^\mu}$ are uniquely determined up to a constant. From the construction of the Dubrovin-Zhang hierarchy (see ~\cite{BPS12b}) it is easy to see that $\left.\frac{\delta\oh_{\beta,d}[u]}{\delta u^\mu}\right|_{u^*_*=0}=0$. We conclude that the local functionals $\oh_{\alpha,d}[u]$ are uniquely determined by $(u^{top})^\alpha$.

We can apply the same arguments to the double ramification hierarchy and the string solution. Since $(u^{str})^\alpha=(u^{top})^\alpha$, the string solution satisfies the same property~\eqref{eq:change of variables}. Note that from the construction of the double ramification hierarchy it is easy to see that $\left.\frac{\delta\og_{\beta,d}[u]}{\delta u^\mu}\right|_{u^*_*=0}=0$. Then we can repeat our arguments and conclude that the local functionals $\og_{\alpha,d}[u]$ are uniquely determined by~$(u^{str})^\alpha$ and, therefore, $\og_{\alpha,d}[u]=\oh_{\alpha,d}[u]$. Theorem~\ref{theorem:dr for cp1} is proved.


{
\appendix

\section{Computations with the extended Toda hierarchy}\label{section:computations}

In this section we prove equations~\eqref{eq:DZ omega Hamiltonian in degree 0},~\eqref{eq:DZ 1 Hamiltonian in degree 0},~\eqref{eq:DZ 1,1 Hamiltonian in degree 0} and~\eqref{eq:DZ omega Hamiltonian}.

For the proofs of equations~\eqref{eq:DZ omega Hamiltonian in degree 0},~\eqref{eq:DZ 1 Hamiltonian in degree 0} and~\eqref{eq:DZ 1,1 Hamiltonian in degree 0}, let us note that from Lemma~\ref{lemma:relation} and formulae~\eqref{eq:S-matrix} it follows that the degree zero parts of the ancestor and the descendant Dubrovin-Zhang hierarchies for~$\CP1$ coincide: $\oh_{\alpha,p}[w]|_{q=0}=\oh_{\alpha,p}^{desc}[w]_{q=0}$. Theorem~\ref{theorem:DZ theorem} says that $\oh_{\alpha,p}^{desc}[w]=\oh_{\alpha,p}^{Td}[w]$. Recall that, by~\eqref{eq:Miura} and~\eqref{eq:w-v relation}, the variables~$u^1,u^\omega$ and $v^1,v^2$ are related in the following way:
\begin{gather}\label{eq:v-u relation}
v^1(u)=e^{\frac{\eps}{2}\d_x}u^1,\qquad v^2(u)=\frac{e^{\frac{\eps}{2}\d_x}-e^{-\frac{\eps}{2}\d_x}}{\eps\d_x}u^\omega.
\end{gather}

\subsection{Proof of~\eqref{eq:DZ omega Hamiltonian in degree 0}}

We have
$$
\oh^{Td}_{\omega,p}[v]|_{q=0}=\int\frac{1}{(p+2)!}\Res\left((e^{\eps\d_x}+v^1)^{p+2}\right)dx=\int\frac{(v^1)^{p+2}}{(p+2)!}dx.
$$
From~\eqref{eq:v-u relation} it follows that
$$
\oh^{Td}_{\omega,p}[u]|_{q=0}=\int\frac{(e^{\frac{\eps}{2}\d_x} u^1)^{p+2}}{(p+2)!}dx=\int e^{\frac{\eps}{2}\d_x}\frac{(u^1)^{p+2}}{(p+2)!}dx=\int\frac{(u^1)^{p+2}}{(p+2)!}dx.
$$
Equation~\eqref{eq:DZ omega Hamiltonian in degree 0} is proved.


\subsection{Proof of~\eqref{eq:DZ 1 Hamiltonian in degree 0}}\label{subsection:proof of DZ Hamiltonian in degree 0}

First of all, let us briefly recall the definition of the logarithm~$\log L$ from~\cite{CDZ04}. The dressing operators $P$ and $Q$:
$$
P=1+\sum_{k\ge 1}p_k e^{-k\eps\d_x},\qquad Q=\sum_{k\ge 0}q_k e^{k\eps\d_x},
$$
are defined by the following identities in the ring of Laurent series in $e^{-\eps\d_x}$ and $e^{\eps\d_x}$ respectively:
$$
L=P\circ e^{\eps\d_x}\circ P^{-1}=Q\circ e^{-\eps\d_x}\circ Q^{-1}.
$$
Note that the coefficients $p_k$ and $q_k$ of the dressing operators do not belong to the ring $\hcA_{v^1,v^2}\otimes\mbC[q,q^{-1}]$ but to a certain extension of it. The logarithm $\log L$ is defined by
$$
\log L:=\frac{1}{2}\left(P\circ\eps\d_x\circ P^{-1}-Q\circ\eps\d_x\circ Q^{-1}\right)=\frac{1}{2}\left(\eps Q_x\circ Q^{-1}-\eps P_x\circ P^{-1}\right).
$$
Here $P_x:=\sum_{k\ge 1}(\d_x p_k) e^{-k\eps\d_x}$ and $Q_x:=\sum_{k\ge 0}(\d_x q_k) e^{k\eps\d_x}$. In~\cite{CDZ04} (see Theorem~2.1) it is proved that the coefficients of~$\log L$ belong to $\hcA^{[0]}_{v^1,v^2}\otimes\mbC[q,q^{-1}]$. 

Actually, we are going to prove a precise formula for the Hamiltonian~$\oh^{Td}_{1,p}[v]|_{q=0}$. From the proof of Theorem~2.1 in~\cite{CDZ04} it follows that the coefficients of the operator $S:=-\eps P_x\circ P^{-1}$ belong to~$\hcA^{[0]}_{v^1,v^2}\otimes\mbC[q]$. Let~$S^0:=S|_{q=0}$. From~\cite{CDZ04} it is also easy to see that the coefficients of $S^0$ belong to $\hcA^{[0]}_{v^1}$. 

For a differential polynomial $f(v^1;\eps)\in\hcA_{v^1}$, let $f^{ev}(v^1;\eps):=\half\left(f(v^1;\eps)+f(v^1;-\eps)\right)$. Following~\cite{CDZ04} let us introduce operators $\cB_+$ and $\cB_-$ by
\begin{gather*}
\cB_+:=\frac{\eps\d_x}{e^{\eps\d_x}-1},\qquad \cB_-:=\frac{\eps\d_x}{1-e^{-\eps\d_x}}.
\end{gather*}
\begin{lemma}\label{lemma:precise formula}
We have 
$$
\oh^{Td}_{1,p}[v]|_{q=0}=\int\left(\frac{(v^1)^{p+1}}{(p+1)!}\cB_-v^2+\frac{2}{(p+1)!}\Res\left((e^{\eps\d_x}+v^1)^{p+1}\circ(S^0-H_{p+1})\right)^{ev}\right)dx.
$$
\end{lemma}  
\begin{proof}
We have
$$
\oh^{Td}_{1,p}[v]=\int\left(\frac{2}{(p+1)!}\Res(L^{p+1}(\log L-H_{p+1}))\right)dx.
$$
Since the coefficients of $L^{p+1}$ belong to~$\hcA^{[0]}_{v^1,v^2}\otimes\mbC[q]$, we obviously have
\begin{gather}\label{eq:formula1}
\left.\Res(L^{p+1})\right|_{q=0}=\Res\left((e^{\eps\d_x}+v^1)^{p+1}\right).
\end{gather}
Let us compute the residue~$\Res(L^{p+1}\log L)|_{q=0}$. We have $\log L=\frac{1}{2}\left(\eps Q_x\circ Q^{-1}-\eps P_x\circ P^{-1}\right)$. As we have already said, the coefficients of $\eps P_x\circ P^{-1}$ belong to~$\hcA^{[0]}_{v^1,v^2}\otimes\mbC[q]$. Therefore,
\begin{gather}\label{eq:formula2}
-\left.\Res\left(L^{p+1}\circ\eps P_x\circ P^{-1}\right)\right|_{q=0}=\Res\left((e^{\eps\d_x}+v^1)^{p+1}\circ S^0\right).
\end{gather}

Let us consider the residue~$\Res\left(L^{p+1}\circ\eps Q_x\circ Q^{-1}\right)$. From the proof of Theorem~2.1 in~\cite{CDZ04} it follows that the coefficients of the operator $\eps Q_x\circ Q^{-1}$ belong to~$\hcA_{v^1,v^2}^{[0]}\otimes\mbC[q,q^{-1}]$. Note that they contain inverse powers of $q$, so we have to be careful while computing the residue~$\left.\Res\left(L^{p+1}\circ\eps Q_x\circ Q^{-1}\right)\right|_{q=0}$. Introduce an operator~$\tQ$ by $\tQ=1+\sum_{k\ge 1}\tq_k e^{k\eps\d_x}:=q_0^{-1}Q$. We have
$$
\tQ\circ e^{-\eps\d_x}\circ\tQ^{-1}=\tL,\quad\text{where}\quad \tL=q_0^{-1}L\circ q_0.
$$ 
The operator $e^{k\eps\d_x}$ can be interpreted as a shift operator, so, following~\cite{CDZ04}, we will sometimes denote $e^{k\eps\d_x}f$ by $f(x+k\eps)$. We have the following identity (see~\cite[eq.~2.21]{CDZ04}):
\begin{gather}\label{eq:identity from cdz}
\frac{q_0(x)}{q_0(x-\eps)}=q e^{v^2}.
\end{gather}
It implies that $\tL=e^{-\eps\d_x}+v^1+q e^{v^2(x+\eps)}e^{\eps\d_x}$. We can compute that
\begin{gather}\label{eq:Q-part}
\eps Q_x\circ Q^{-1}=\eps q_0(q_0)^{-1}+q_0\eps\tQ_x\circ\tQ^{-1}\circ q_0^{-1}.
\end{gather}
We have (\cite[eq.~2.21]{CDZ04}) $\frac{(q_0)_x}{q_0}-\frac{(q_0)_x(x-\eps)}{q_0(x-\eps)}=v^2_x$. Therefore, $\eps q_0(q_0)^{-1}=\frac{\eps\d_x}{1-e^{-\eps\d_x}}v^2=\cB_-v^2$. We see that~$\Res(L^{p+1}\circ\eps q_0 q_0^{-1})\in\hcA^{[0]}_{v^1,v^2}\otimes\mbC[q]$ and
\begin{gather}\label{eq:formula3}
\Res(L^{p+1}\circ\eps q_0 q_0^{-1})|_{q=0}=\Res\left((e^{\eps\d_x}+v^1)^{p+1}\circ\eps q_0 q_0^{-1}\right)=\frac{(v^1)^{p+1}}{(p+1)!}\cB_-v^2.
\end{gather}

It remains to compute~$\Res(L^{p+1}\circ q_0\eps\tQ_x\circ\tQ^{-1}\circ q_0^{-1})$. We have
$$
\Res\left(L^{p+1}\circ q_0\eps\tQ_x\circ\tQ^{-1}\circ q_0^{-1}\right)=\Res\left(\tL^{p+1}\circ\eps\tQ_x\circ\tQ^{-1}\right).
$$
It is easy to relate the operator~$\eps \tQ_x\circ \tQ^{-1}$ to the operator~$S=-\eps P_x\circ P^{-1}$. Note that 
$$
\tL=\left.L\right|_{\substack{v^2(x)\mapsto v^2(x-\eps)\\\eps\mapsto-\eps}}.
$$
Therefore,
$$
\eps\tQ_x\circ\tQ^{-1}=\left.S\right|_{\substack{v^2(x)\mapsto v^2(x-\eps)\\\eps\mapsto-\eps}}.
$$
We immediately see that the coefficients of $\eps\tQ_x\circ\tQ^{-1}$ belong to $\hcA_{v^1,v^2}^{[0]}\otimes\mbC[q]$. We get
\begin{multline}\label{eq:formula4}
\left.\Res\left(\tL^{p+1}\circ\eps\tQ_x\circ\tQ^{-1}\right)\right|_{q=0}=\Res\left((e^{-\eps\d_x}+v^1)^{p+1}\circ\left.\left(\eps\tQ_x\circ\tQ^{-1}\right)\right|_{q=0}\right)=\\
=\Res\left((e^{-\eps\d_x}+v^1)^{p+1}\circ\left.S^0\right|_{\eps\mapsto-\eps}\right)=\left.\Res\left((e^{\eps\d_x}+v^1)^{p+1}\circ S^0\right)\right|_{\eps\mapsto-\eps}.
\end{multline}
Collecting equations~\eqref{eq:formula1},~\eqref{eq:formula2},~\eqref{eq:formula3} and~\eqref{eq:formula4} we get the statement of the lemma.
\end{proof}
Let us prove equation~\eqref{eq:DZ 1 Hamiltonian in degree 0}. Using the proof of Theorem~2.1 in~\cite{CDZ04} it is easy to compute that $S^0=\sum_{k\ge 1}f_ke^{-k\eps\d_x}$, where $f_k=\frac{(-1)^{k-1}}{k}(v^1)^k+O(\eps)$. Therefore, we have
$$
\left.\Res\left((e^{\eps\d_x}+v^1)^{p+1}\circ S^0\right)\right|_{\eps=0}=\left(\sum_{i=1}^{p+1}{p+1\choose i}\frac{(-1)^{i-1}}{i}\right)(v^1)^{p+1}.
$$
Let us denote the sum on the right-hand side by $C_{p+1}$. For $k\ge 1$, we have
$$
C_{k+1}-C_k=\sum_{i=0}^k{k\choose i}\frac{(-1)^i}{i+1}=\sum_{i=0}^k{k\choose i}(-1)^i\int_0^1 x^i dx=\int_0^1(1-x)^k dx=\frac{1}{k+1}.
$$
Since $C_1=1$, we obtain $C_{p+1}=H_{p+1}$. We get
$$
\left.\Res\left((e^{\eps\d_x}+v^1)^{p+1}\circ(S^0-H_{p+1}\right)\right|_{\eps=0}=0.
$$
Using Lemma~\ref{lemma:precise formula} we can conclude that
$$
\left.h^{Td}_{1,p}[v]\right|_{q=0}=\int\left(\frac{(v^1)^{p+1}}{(p+1)!}\cB_- v^2+\sum_{g\ge 1}\eps^{2g} r_{p,g}(v^1)\right)dx,
$$
where $r_{p,g}\in\mcA_{v^1}$, $\deg_{dif} r_{p,g}=2g$. Finally, equation~\eqref{eq:v-u relation} implies that
$$
\left.h^{Td}_{1,p}[u]\right|_{q=0}=\int\left(\frac{(u^1)^{p+1}}{(p+1)!}u^\omega+\sum_{g\ge 1}\eps^{2g} r_{p,g}(u^1)\right)dx.
$$
Equation~\eqref{eq:DZ 1 Hamiltonian in degree 0} is proved.


\subsection{Proof of~\eqref{eq:DZ 1,1 Hamiltonian in degree 0}} 

Let us use Lemma~\ref{lemma:precise formula}. It is not hard to compute the first two coefficients of the operator $S^0$:
$$
S^0=\cB_+v^1 e^{-\eps\d_x}+\left(\frac{1}{2}\cB_+(v^1)^2-v^1(x-\eps)\cB_+ v^1\right)e^{-2\eps\d_x}+\ldots.
$$
Therefore, we have
$$
\Res\left((e^{\eps\d_x}+v^1)^2\circ S^0\right)=e^{2\eps\d_x}\left(\frac{1}{2}\cB_+(v^1)^2-v^1(x-\eps)\cB_+ v^1\right)+(v^1(x)+v^1(x+\eps))e^{\eps\d_x}\cB_+ v^1.
$$
Taking the integral we get
\begin{multline*}
\int\Res\left((e^{\eps\d_x}+v^1)^2\circ S^0\right)dx=\int\left(\frac{(v^1)^2}{2}-v^1(x-\eps)\cB_+v^1+(v^1+v^1(x+\eps))e^{\eps\d_x}\cB_+ v^1\right)dx=\\
=\int\left(\frac{(v^1)^2}{2}+v^1\cB_+ v^1\right)dx.
\end{multline*}
Let us take the even part:
\begin{multline*}
\int\left(\frac{(v^1)^2}{2}+v^1\cB_+ v^1\right)^{ev}dx=\int\left(\frac{1}{2}(v^1)^2+\frac{1}{2}v^1\eps\d_x\frac{e^{\frac{\eps}{2}\d_x}+e^{-\frac{\eps}{2}\d_x}}{e^{\frac{\eps}{2}\d_x}-e^{-\frac{\eps}{2}\d_x}} v^1\right)dx=\\
=\int\left(\frac{3}{2}(v^1)^2+\sum_{g\ge 1}\eps^{2g}v^1\frac{B_{2g}}{(2g)!}v^1_{2g}\right)dx.
\end{multline*}
We get $\left.h^{Td}_{1,1}[v]\right|_{q=0}=\int\left(\frac{(v^1)^2}{2}\cB_-v^2+\sum_{g\ge 1}\eps^{2g}v^1\frac{B_{2g}}{(2g)!}v^1_{2g}\right)dx$ and, therefore, 
$$
\left.h^{Td}_{1,1}[u]\right|_{q=0}=\int\left(\frac{(u^1)^2}{2}u^\omega+\sum_{g\ge 1}\eps^{2g}u^1\frac{B_{2g}}{(2g)!}u^1_{2g}\right)dx.
$$
Equation~\eqref{eq:DZ 1,1 Hamiltonian in degree 0} is proved.


\subsection{Proof of~\eqref{eq:DZ omega Hamiltonian}} We have
$$
\oh_{\omega,0}^{Td}[v]=\int\frac{1}{2}\Res\left(\left(e^{\eps\d_x}+v^1+qe^{v^2}e^{-\eps\d_x}\right)^2\right)dx=\int\left(\frac{(v^1)^2}{2}+qe^{v^2}\right)dx.
$$
Using~\eqref{eq:v-u relation} we get $\oh^{Td}_{\omega,0}[u]=\int\left(\frac{(u^1)^2}{2}+qe^{S(\eps\d_x)u^\omega}\right)dx$. Finally, if we apply Lemma~\ref{lemma:relation}, we obtain
$$
\oh_{\omega,0}[u]=\int\left(\frac{(u^1)^2}{2}+q\left(e^{S(\eps\d_x)u^\omega}-u^\omega\right)\right)dx.
$$
Equation~\eqref{eq:DZ omega Hamiltonian} is proved.

}


\begin{thebibliography}{BSSZ12}

\bibitem[Bur15a]{Bur14} A. Buryak, {\it Double ramification cycles and integrable hierarchies}, Communications in Mathematical Physics~336 (2015), no. 3, 1085-1107.

\bibitem[Bur15b]{Bur13} A. Buryak, {\it Dubrovin-Zhang hierarchy for the Hodge integrals}, Communications in Number Theory and Physics~9 (2015), no. 2, 239-271.

\bibitem[BPS12a]{BPS12a} A. Buryak, H. Posthuma, S. Shadrin, {\it On deformations of quasi-Miura transformations and the Dubrovin-Zhang bracket}, Journal of Geometry and Physics~62 (2012), no. 7, 1639-1651.

\bibitem[BPS12b]{BPS12b} A. Buryak, H. Posthuma, S. Shadrin, {\it A polynomial bracket for the Dubrovin-Zhang hierarchies}, Journal of Differential Geometry~92 (2012), no. 1, 153-185.

\bibitem[BSSZ12]{BSSZ12} A. Buryak, S. Shadrin, L. Spitz, D. Zvonkine, {\it Integrals of psi-classes over double ramification cycles}, American Journal of Mathematics~137 (2015), no. 3, 699-737.

\bibitem[CDZ04]{CDZ04} G. Carlet, B. Dubrovin, Y. Zhang, {\it The extended Toda hierarchy}, Moscow Mathematical Journal~4 (2004), no.~2, 313-332.

\bibitem[CMW12]{CMW12} R. Cavalieri, S. Marcus, J. Wise, {\it Polynomial families of tautological classes on $\mathcal{M}_{g,n}^{rt}$}, Journal of Pure and Applied Algebra~216 (2012), no.~4, 950-981.

\bibitem[Ch06]{Ch06} A. Chiodo, {\it The Witten top Chern class via K-theory}, Journal of Algebraic Geometry~15 (2006), no.~4, 681-707.

\bibitem[Dub96]{Dub96} B. Dubrovin, {\it Geometry of 2D topological field theory}, in {\it Integrable Systems and Quantum Groups}, Lecture Notes in Mathematics, Volume 1620, 1996, 120-348.

\bibitem[DZ04]{DZ04} B. Dubrovin, Y. Zhang, {\it Virasoro symmetries of the extended Toda hierarchy}, Communications in Mathematical Physics~250 (2004), no.~1, 161-193. 

\bibitem[DZ05]{DZ05} B.~A.~Dubrovin, Y.~Zhang, {\it Normal forms of hierarchies of integrable PDEs, Frobenius manifolds and Gromov-Witten invariants}, a new 2005 version of arXiv:math/0108160v1, 295 pp.

\bibitem[EGH00]{EGH00} Y. Eliashberg, A. Givental and H. Hofer, {\it Introduction to symplectic field theory},  GAFA 2000 Visions in Mathematics special volume, part II, 560-673, 2000.

\bibitem[FP03]{FP03} C. Faber, R. Pandharipande, {\it Hodge integrals, partition matrices, and the $\lambda_g$-conjecture}, Annals of Mathematics~157 (2003), no.~1, 97-124.

\bibitem[FSZ10]{FSZ10} C. Faber, S. Shadrin, D. Zvonkine, {\it Tautological relations and the $r$-spin Witten conjecture}, Annales Scientifiques de l'Ecole Normale Superieure (4) 43 (2010), no. 4, 621-658.

\bibitem[FR11]{FR10} O.~Fabert, P.~Rossi, {\it String, dilaton and divisor equation in Symplectic Field Theory}, International Mathematics Research Notices IMRN 2011, no. 19, 4384-4404.

\bibitem[GD76]{GD76} I. M. Gelfand, L. A. Dikii, {\it Fractional powers of operators and Hamiltonian systems}, Functional Analysis and its Applications~10 (1976), 259-273.

\bibitem[Get98]{Get97} E. Getzler, {\it Topological recursion relations in genus 2}, in {\it Integrable systems and algebraic geometry} (Kobe/Kyoto 1997), World Scientific Publishing: River Edge, NJ 1998, 73-106.

\bibitem[Get99]{Get98} E. Getzler, {\it The Virasoro conjecture for Gromov-Witten invariants}, Algebraic geometry: Hirzebruch 70 (Warsaw, 1998), 147-176, Contemp. Math., 241, Amer. Math. Soc., Providence, RI, 1999.

\bibitem[GP98]{GP98} E. Getzler, R. Pandharipande, {\it Virasoro constraints and the Chern classes of the Hodge bundle}, Nuclear Physics~B~530 (1998), no.~3, 701-714.

\bibitem[Giv01]{Giv01} A. Givental, {\it Gromov-Witten invariants and quantization of quadratic Hamiltonians}, Moscow Mathematical Journal~1 (2001), no. 4, 551-568. 

\bibitem[Hai13]{H11} R.~Hain, \emph{Normal functions and the geometry of moduli spaces of curves}, Handbook of moduli, Vol. I, 527-578, Adv. Lect. Math. (ALM), 24, Int. Press, Somerville, MA, 2013.

\bibitem[Hor95]{Hor95} K. Hori, {\it Constraints for topological strings in $D\ge 1$}, Nuclear Physics B~439 (1995), no.~1-2, 395-420. 

\bibitem[KM94]{KM94} M. Kontsevich, Yu. Manin, {\it Gromov-Witten classes, quantum cohomology, and enumerative geometry}, Communications in Mathematical Physics~164 (1994), no. 3, 525-562.

\bibitem[MW13]{MW13} S. Marcus, J. Wise, {\it Stable maps to rational curves and the relative Jacobian}, arXiv:1310.5981. 

\bibitem[Pan00]{Pan00} R. Pandharipande, {\it The Toda equations and the Gromov-Witten theory of the Riemann sphere}, Letters in Mathematical Physics~53 (2000), no.~1, 59-74. 

\bibitem[PPZ15]{PPZ13} R. Pandharipande, A. Pixton, D. Zvonkine, {\it Relations on $\oM_{g,n}$ via $3$-spin structures}, Journal of the American Mathematical Society~28 (2015), no. 1, 279-309.

\bibitem[PV01]{PV00} A. Polishchuk and A. Vaintrob, {\it Algebraic construction of Witten's top Chern class}, Advances in algebraic geometry motivated by physics (Lowell, MA, 2000), 229-249, Contemp. Math., 276, Amer. Math. Soc., Providence, RI, 2001.

\bibitem[Ros10]{Ros09} P. Rossi, {\it Integrable systems and holomorphic curves}, Proceedings of the G\"okova Geometry-Topology Conference 2009, 34-57, Int. Press, Somerville, MA, 2010.

\bibitem[Ros12]{Ros12} P. Rossi, {\it Nijenhuis operator in contact homology and descendant recursion in symplectic field theory}, Proceedings of the G\"okova Geometry-Topology Conference 2014, International Press, May 2015, arXiv:1201.1127.

\bibitem[SAK79]{SAK79} J. Satsuma, M. J. Ablowitz, Y. Kodama, {\it On an internal wave equation describing a stratified fluid with finite depth}, Physics Letters A 73 (1979), no. 4, 283-286.

\bibitem[Wi93]{Wi93} E. Witten, {\it Algebraic geometry associated with matrix models of two-dimensional gravity}, in {\it Topological methods in modern mathematics} (Stony Brook, NY, 1991), 235-269, Publish or Perish, Houston, TX, 1993.

\end{thebibliography}
\end{document}